\PassOptionsToPackage{final}{graphicx}
\PassOptionsToPackage{nosumlimits,nonamelimits}{amsmath}
\PassOptionsToPackage{capitalize,noabbrev,nameinlink}{cleveref}
\PassOptionsToPackage{colorlinks,linkcolor={blue},citecolor={blue},urlcolor={blue},breaklinks=true,final}{hyperref}
\documentclass[
a4paper,
UKenglish,
cleveref,
numberwithinsect,
nolipics
]{mylipics-v2021}

\nolinenumbers
\hideLIPIcs
\title{Non-Expansive Fuzzy Modal Logic\\ Has Uniform Interpolation}
\titlerunning{Non-Expansive Fuzzy Modal Logic Has Uniform Interpolation}
 \author{Stefan Gebhart}{Friedrich-Alexander-University Erlangen-N\"urnberg, Germany}{stefan.gebhart@fau.de}{}{acknowledges support by the Deutsche Forschungsgemeinschaft (DFG, German Research Foundation) -- project number 531706730}

 \author{Lutz Schr\"oder}{Friedrich-Alexander-University Erlangen-N\"urnberg, Germany}{}{}{}
 \author{Paul Wild}{Friedrich-Alexander-University Erlangen-N\"urnberg, Germany}{}{}{}

 \authorrunning{S. Gebhart, L. Schr\"oder and P. Wild}

 \Copyright{Stefan Gebhart, Lutz Schr\"oder and Paul Wild}

\usepackage{amsmath}
\usepackage{amssymb}
\usepackage{amsthm}
\usepackage{xspace}
\usepackage{stmaryrd}

\usepackage{mathtools}
\usepackage{enumitem}
\usepackage{booktabs}

\usepackage{mathpartir} 

\usepackage[layout=footnote,marginclue,final]{fixme}
\FXRegisterAuthor{sg}{asg}{SG}	
\FXRegisterAuthor{ls}{als}{LS}	
\FXRegisterAuthor{pw}{apw}{PW}	

\theoremstyle{definition}
\newtheorem{defn}[theorem]{Definition}
\crefname{defn}{definition}{definitions}
\Crefname{defn}{Definition}{Definitions}
\newtheorem{convention}[theorem]{Convention}

\newcommand{\TwoEXP}{\textsc{2ExpTime}\xspace}
\newcommand{\EXP}{\textsc{ExpTime}\xspace}
\newcommand{\NEXP}{\textsc{NExpTime}\xspace}
\newcommand{\PSPACE}{\textsc{PSpace}\xspace}

\newcommand{\ALC}{\mathcal{ALC}}

\newcommand{\Sem}[1]{\llbracket #1\rrbracket}

\newcommand{\ineqOp}{\bowtie}
\newcommand{\ineqOpSymbol}{\mathord{\ineqOp}}
\newcommand{\ineqOpSet}{\{\geq, >, \leq, <\}}

\newcommand{\gineqOp}{\triangleright}
\newcommand{\gineqOpSymbol}{\mathord{\gineqOp}}
\newcommand{\gineqOpSet}{\{\geq, >\}}

\newcommand{\lineqOp}{\triangleleft}
\newcommand{\lineqOpSymbol}{\mathord{\lineqOp}}
\newcommand{\lineqOpSet}{\{\leq, <\}}

\newcommand{\covcompDef}[2]{\mathord{#1}\mathbin{\star}\mathord{#2}}
\newcommand{\covcomp}[2]{\mathrel{(\covcompDef{#1}{#2})}}

\newcommand{\unitInterval}{[0,1]}

\newcommand{\strictnessOp}[1]{\mathrel{\overline{\mathord{#1}}}}
\newcommand{\directionOp}[1]{\mathrel{\mathord{#1}^{-1}}}
\newcommand{\complementOp}[1]{\mathrel{\mathord{#1}^{c}}}

\newcommand{\voc}{\operatorname{\mathsf{Voc}}}

\newcommand{\atoms}{\mathsf{At}}

\newcommand{\lit}[1]{\mathsf{Lit}_{\neq #1}}

\newcommand{\elim}[1]{\operatorname{\mathsf{Elim}}_{#1}}

\newcommand{\normassertiondefault}{N}
\newcommand{\normassertion}[2]{\normassertiondefault_{(\mathord{#1},#2)}}

\newcommand{\boxsequents}[1]{#1^\Box}
\newcommand{\diamondsequents}[1]{#1^\Diamond}

\newcommand{\boxsequentsresnormalform}[2]{\operatorname{B}_{#2}^{#1}}
\newcommand{\diamondsequentsresnormalform}[2]{\operatorname{D}_{#2}^{#1}}

\newcommand{\interpolationCalculus}{interpolation calculus\xspace}
\newcommand{\InterpolationCalculus}{Interpolation calculus\xspace}

\newcommand{\sequenttreename}{computation tree\xspace}

\newcommand{\sequentlabelname}{$p$-configuration\xspace}
\newcommand{\sequentlabelnameplural}{$p$-configurations\xspace}

\newcommand{\ruleSize}{\normalsize}

\newcommand{\symbolicSequent}{strict sequent}

\newcommand{\clip}[1]{\left[#1\right]_0^1}
\newcommand{\shiftformula}[2]{S_{#1}\!\left(#2\right)}

\newcommand{\boxsequentsresnormalformsymbolic}[2]{\operatorname{Box}_{#2}^{#1}}
\newcommand{\diamondsequentsresnormalformsymbolic}[2]{\operatorname{Dia}_{#2}^{#1}}

\begin{document}

\begin{titlepage}
  \maketitle

  \begin{abstract}
    Fuzzy logics with truth values in the unit interval vary in their
    choice of propositional connectives and their semantics, ranging
    from just minimum and maximum as in Zadeh logics to truncated
    addition and subtraction as in \L{}ukasiewicz logics or their
    extensions with truth constants, known as Pavelka logics. In the
    more specific world of fuzzy modal logics, non-expansive fuzzy modal
    logic has recently been proposed as a sweet spot between
    \L{}ukasiewicz (more precisely, Pavelka) and Zadeh fuzzy modal
    logic; it restricts \L{}ukasiewicz propositional connectives to have
    one argument constant, equivalently extends the Zadeh propositional
    base with constant shift operators. It retains reasonable
    expressiveness both in practical terms and abstractly; in
    particular, it characterizes the standard notion of behavioural
    distance on fuzzy Kripke frames. At the same time, it allows for
    satisfiability checking in the same complexity as for its two-valued
    correspondent (the modal logic~$K$, or the description
    logic~$\ALC$), in contrast to the notorious computational hardness
    of \L{}ukasiewicz modal logic. In the present work, we show that
    non-expansive fuzzy modal logic admits uniform interpolation, a key
    property in modularization, information hiding, summarization, and
    other applications that fails for both \L{}ukasiewicz (or Pavelka)
    and Zadeh modal logic.
  \end{abstract}
\end{titlepage}
\newpage
\section{Introduction}

\emph{Fuzzy description logics} (see~\cite{LukasiewiczStraccia08} for
an overview) serve the specification of vague knowledge, expressing
both vague properties of individuals such as tallness and vague
relations such as fandom. Fuzzy description logics are equivalent to
the more neutrally termed \emph{fuzzy modal logics}, which may more
generally be seen as logics for weighted transition systems, and as
such have found application, for instance, in software model
checking~\cite{BrunsGodefroid04,KupfermanLustig07}. The semantics of
fuzzy modal logics varies widely even when truth values are fixed to
the unit interval, in particular regarding the interpretation of
propositional connectives~\cite{LukasiewiczStraccia08}, with
substantial impact on logical properties and the complexity of
reasoning. At one end of the spectrum, Zadeh semantics, which
essentially features only minimum, maximum, and fuzzy negation,
ensures tractability even of very expressive description
logics~\cite{Straccia01,StoilosEA07,StoilosEA14} but on the other hand
is semantically almost indistinguishable from its classical
counterpart (e.g.~for $c\ge 0.5$, a concept in Zadeh fuzzy~$\ALC$ is
satisfiable with degree at least~$c$ iff the same concept is
satisfiable in
classical~$\ALC$)~\cite{BonattiTettamanzi03,KellerHeymans09,Straccia01}. At
the other end, \L{}ukasiewicz semantics, which additionally makes use
of the additive structure of the unit interval, is characterized by
desirable logical properties~\cite{KunduChen98} but causes high
complexity of the main reasoning problems. For instance, the best
known upper bound for satisfiability checking in \L{}ukasiewicz fuzzy
$\ALC$ is \NEXP~\cite{Straccia05,KulackaEA13}, and in the presence of
global assumptions (i.e.~general TBoxes, in description logic
terminology), satisfiability checking in fact becomes
undecidable~\cite{BaaderPenaloza11}.

As an intermediate logic aiming to strike a balance between
expressiveness and computational tractability, \emph{non-expansive
  fuzzy modal logic} or \emph{non-expansive fuzzy $\ALC$} has been
proposed~\cite{WildEA18,GebhartEA25}. Here, the term
\emph{non-expansive} refers to the fact that formula evaluation is
non-expansive w.r.t.~the natural behavioural distance on fuzzy
transition systems. In fact, non-expansive fuzzy modal logic is
expressive for this distance in two senses: It satisfies both a
Hennessy-Milner-type theorem saying that logical distance coincides
with behavioural distance on finitely-branching systems, and a
van-Benthem-type theorem saying that, up to approximation in bounded
rank, fuzzy modal logic is precisely the non-expansive fragment of its
natural fuzzy first-order correspondence
language~\cite{WildEA18}. Syntactically, the logic lies between Zadeh
fuzzy modal logic and the extension of \L{}ukasiewicz fuzzy modal
logic with (rational) truth constants, a logic known as (rational)
Pavelka logic~\cite{Hajek95,Pavelka79}. Specifically, what is termed
the \emph{non-expansive propositional base} restricts \L{}ukasiewicz
connectives to have a truth constant as one argument or, equivalently,
extends the Zadeh base with constant shift operators (this
propositional base was first used in quantitative modal logics
for probabilistic transition systems~\cite{BreugelWorrell05}). For
instance, this allows expressing dampening effects in the inheritance
of properties, e.g.~in axioms saying that friends of football fanatics
are also football fanatics, maybe to a slightly lesser
degree~\cite{GebhartEA25}. While the expressiveness of non-expansive
fuzzy modal logic thus markedly exceeds that of Zadeh modal logic, the
complexity of the main reasoning problems is not increased,
i.e.~remains $\PSPACE$ for plain satisfiability and $\EXP$ for
satisfiability under global assumptions~\cite{GebhartEA25} (the same
as for classical $\ALC$).

In the present work, we establish \emph{uniform interpolation} for
non-expansive fuzzy modal logic, to our knowledge the first
interpolation result in fuzzy modal logic. Generally, a uniform
(post-)interpolant of a formula w.r.t.~a propositional atom~$p$ is,
roughly speaking, a strongest $p$-free consequence of the original
formula, i.e.~captures the content of the original formula in the
reduced vocabulary obtained by removing~$p$. Existence of uniform
interpolants implies \emph{Craig interpolation}, which says that every
valid implication is mediated by a formula mentioning only the shared
vocabulary of the two formulae.  While interpolation is generally a
subtle issue and not stable under extension or restriction of logics,
it is known that the basic two-valued modal logic~$\mathbf{K}$ has
uniform interpolation~\cite{Visser96}. Applications of Craig
interpolation and uniform interpolation in classical modal and
description logics include summarization, analysis, modularization,
and reuse of knowledge
bases~\cite{KonevEA09,LutzWolter11,KoopmannSchmidt15}, information
hiding~\cite{Grau10}, knowledge exchange between agents
(e.g.~\cite{ToluhiEA21}), the computation of separating formulae in
supervised concept learning~\cite{JungEA21}, and the computation of
explicit definitions~\cite{LyonKarge24}.

As the central step towards our main result, we establish a uniform
interpolation result for \emph{threshold judgements} stating that the
truth value of a given formula is at least a given
constant. Technically, we base the computation of uniform interpolants
in this sense on a sequent calculus for threshold judgements that
dualizes a previous tableau calculus~\cite{GebhartEA25}. We obtain an
exponential bound on the size of uniform interpolants, the same as in
the classical case~\cite{TenCateEA06,TenCateEA26}. By introducing a
symbolic treatment of thresholds, we subsequently extend the
construction to obtain uniform interpolation of formulae in the
expected sense (that is, w.r.t.~the pointwise order on truth
functions), with the same exponential size bound.  We note that this is
in sharp contrast to both Zadeh semantics and full \L{}ukasiewicz
semantics, under which even basic propositional logic fails to admit
even Craig interpolation~\cite{BaazVeith99} (propositional
\L{}ukasiewicz logic does have the weaker property of deductive
interpolation~\cite{DiNolaLettieri00}); we show in the appendix that
this extends to Pavelka logic.

\subparagraph*{Organization} Basic concepts of the syntax and
semantics of fuzzy (propositional) logics are discussed in
\cref{sec:prelim}. We recall the definition of non-expansive fuzzy
modal logic in \cref{sec:fuzzy-ml}. The sequent calculus, discussed in
\cref{sec:sequent}, relies on several normalization procedures,
introduced in \cref{normalforms}. The algorithm for the computation of
uniform interpolants of judgements is presented in
\cref{sec:uniform-interpolation}, and the main result on uniform
interpolation of formulae is concluded in
\cref{sec:pointwise-interpolation}.

\subparagraph*{Related work} The proof-theoretic approach to computing
uniform interpolants goes back to work on intuitionistic propositional
logic~\cite{Pitts92} and has previously been adapted to two-valued
modal logics~\cite{Bilkova07}. Forthcoming surveys are available of
work on interpolation and its applications in non-classical
logics~\cite{Fussner26} and two-valued description
logics~\cite{JungEA26}. We note that uniform interpolation of TBoxes
in two-valued $\ALC$ is conceptually and computationally harder than
interpolation of concepts (or formulae)~\cite{LutzWolter11}: A given
TBox may fail to have a uniform interpolant; deciding whether a TBox
has a uniform interpolant is \TwoEXP-complete; and the size of uniform
interpolants may be triply exponential. While it is known that the
modal logics $\mathbf{K}$, $\mathbf{T}$, $\mathbf{GL}$ (Gödel-Löb's
logic of provability), and Grzegorczyk's $\mathbf{S4Grz}$ have uniform
interpolation~\cite{Visser96,Bilkova07}, uniform interpolation fails
for $\mathbf{S4}$~\cite{GhilardiZawadowski95}
and~$\mathbf{K4}$~\cite{Bilkova07}. Recent work establishes uniform
interpolation for the constructive modal logics~$\mathbf{CK}$ and
$\mathbf{WK}$~\cite{vanderGiessenShillito26} and on the other hand
failure of even Craig interpolation for Fischer-Servi intuitionistic
modal logic~$\mathbf{IK}$~\cite{AlmeidaEA26}. Craig interpolation in
extensions of Belnap-Dunn bimodal logic (a paraconsistent modal logic
based on a four-element lattice of truth values) has been shown to
reduce to Craig interpolation in extensions of two-valued bimodal
logic~\cite{Speranski22}. As indicated above, no positive results
appear to be known on interpolation in fuzzy modal
logics. Interpolation in the sense used here is unrelated to
\emph{fuzzy rule interpolation}~\cite{KoczyHirota93}, in which the
term `interpolation' refers to standard interpolation between real
numbers.

\section{Preliminaries on Fuzzy Logic and
  Interpolation}\label{sec:prelim}

We briefly recall basic definitions and facts in propositional fuzzy
logics. For our present purposes, we understand fuzzy logics as taking
truth values in the unit interval, which in particular is known to be
complete for \L{}ukasiewicz logic (in the sense that any formula that
is valid in the unit interval is valid in all so-called
MV-algebras)~\cite{Chang59}. We equip the unit interval with
\emph{truncated addition}~$\oplus$ and \emph{truncated
  subtraction}~$\ominus$, defined by $x\oplus y=\min(x+y,1)$ and
$x\ominus y=\max(x-y,0)$.

As indicated in the introduction, fuzzy logics differ in their
interpretation of the propositional connectives. For purposes of the
discussion of propositional logics, we work with varying syntax that
includes in particular a set~$V$ of \emph{propositional atoms}; the
\emph{vocabulary} of a formula~$\phi$ is then the set of propositional
atoms that occur in~$\phi$. The propositional semantics is based on
valuations $\kappa\colon V\to[0,1]$, and assigns to each propositional
formula~$\phi$ built from~$V$ using various propositional connectives
and each valuation~$\kappa$ a truth value
$\Sem{\phi}\kappa\in[0,1]$. A formula~$\phi$ is \emph{valid} if
$\Sem{\phi}\kappa=1$ for all~$\kappa$. In \emph{Zadeh logic},
conjunction~$\land$ and disjunction~$\lor$ are just interpreted as
minimum and maximum,
i.e.~$\Sem{\phi\land\psi}\kappa=\min(\Sem{\phi}\kappa,\Sem{\psi}\kappa)$
and
$\Sem{\phi\lor\psi}\kappa=\max(\Sem{\phi}\kappa,\Sem{\psi}\kappa)$,
while negation~$\neg$ is interpreted by
$\Sem{\neg\phi}\kappa=1-\Sem{\phi}\kappa$, an operation known briefly as
\emph{fuzzy negation}. In the sequel, we often just write
$x\vee y=\max(x,y)$ and $x\wedge y=\min(x,y)$ for $x,y\in[0,1]$, and
similarly $\bigvee_{i\in I} x_i=\sup_{i\in I}x_i$,
$\bigwedge_{i\in I}x_i=\inf_{i\in I}x_i$ when $x_i\in[0,1]$ for all
$i\in I$. \emph{\L{}ukasiewicz logic} additionally features an
additive disjunction operator~$\oplus$ interpreted as truncated
addition,
$\Sem{\phi\oplus\psi}\kappa=\Sem{\phi}\kappa\oplus\Sem{\psi}\kappa$. The
dual operator~$\otimes$, given by
$\phi\otimes\psi\equiv\neg(\neg\phi\oplus\neg\psi)$, is known as
\emph{strong conjunction} or \emph{\L{}ukasiewicz conjunction}, and
explicitly given by
$\Sem{\phi\otimes\psi}\kappa=\max(\Sem{\phi}\kappa+\Sem{\psi}\kappa-1,0)$. (In
fact, $\land$ and $\lor$ are definable in terms of~$\oplus$
and~$\neg$.) The extension of \L{}ukasiewicz logic with (rational)
truth constants is known as \emph{(rational) Pavelka
  logic}~\cite{Pavelka79}.

Fuzzy logics can be equipped with different notions of logical
consequence. One frequently considered notion is to take~$\psi$ to be
a logical consequence of~$\phi$ if whenever~$\phi$ takes the truth
value~$1$, then so does~$\psi$. For distinction, we then say
that~$\psi$ is a \emph{weak consequence} of~$\phi$. This is the notion
of logical consequence that is meant in the mentioned statement that
propositional \L{}ukasiewicz logic has the \emph{deductive
  interpolation property}~\cite{DiNolaLettieri00}, which says that
every logical consequence between two formulae can be mediated by a
formula mentioning only the shared vocabulary. (In
\cref{sec:uniform-interpolation}, we establish a generalized form of
deductive uniform interpolation as an intermediate result.)
Contrastingly, the standard definition of \emph{Craig
  interpolation} says that every valid \emph{implication} between two
formulae can be mediated by a formula mentioning only the shared
vocabulary, a stronger property than deductive
interpolation~\cite{Fussner26}. In \L{}ukasiewicz logic, with
implication~$\phi\to\psi$ defined as $\neg\phi\oplus\psi$, the
implication $\phi\to\psi$ is valid iff
$\Sem{\phi}\kappa\le\Sem{\psi}\kappa$ for all~$\kappa$. To apply this
notion to logics that do not have \L{}ukasiewicz implication, we
denote the latter property (and generalizations thereof to more
expressive logics) as $\models\phi\leq\psi$, and refer to this form of
logical consequence as \emph{strong consequence}. As a simple
separating example, in \L{}ukasiewicz logic, $\phi\otimes\phi$ is a
weak consequence of~$\phi$ but not a strong consequence. From now on,
we refer to interpolation w.r.t.~strong consequence, i.e.~the property
that every strong consequence between two formulae is mediated by a
formula mentioning only the shared vocabulary, as \emph{Craig
  interpolation}. It is well-known that both \L{}ukasiewicz
logic~\cite{BaazVeith99} and Zadeh logic fail to have Craig
interpolation in this sense, the obvious counterexample being the
strong consequence $\models p\land\neg p\le q\lor\neg q$. Indeed, a
putative interpolant would not mention any propositional atoms, and
hence would need to have constant value $0.5$, which is clearly not
expressible in either Zadeh or \L{}ukasiewicz logic. We show in the
appendix that the failure of Craig interpolation extends to Pavelka
logic, with a slightly more complicated
counterexample. Contrastingly, our main result
(\cref{theorem:pointwise-interpolation}) establishes Craig and even
uniform interpolation w.r.t.~strong consequence for non-expansive
fuzzy modal logic.

\section{Non-Expansive Fuzzy Modal Logic}\label{sec:fuzzy-ml}
We recall the syntax and semantics of \emph{non-expansive fuzzy modal
  logic}~\cite{WildEA18} (whose multi-modal version has been termed
\emph{non-expansive fuzzy $\ALC$}~\cite{GebhartEA25}; we restrict to
the uni-modal case purely in the interest of readability), and fix
notation used in sequents; a corresponding sequent calculus is
discussed in \cref{sec:sequent}. The logic is characterized by
employing a \emph{non-expansive propositional base} that, in the
terminology recalled in \cref{sec:prelim}, lies between the Zadeh base
and the Pavelka base.

\subparagraph*{Syntax} We fix a set~$\atoms$ of \emph{(propositional)
  atoms}. The \emph{formulae} $\phi,\psi,\dots$ of \emph{non-expansive
  fuzzy modal logic} are then generated by the grammar
\begin{equation}\label{eq:grammar}
  \phi,\psi \Coloneqq p \mid 0 \mid \lnot \phi \mid \phi \ominus c \mid \phi\land\psi \mid \Diamond\phi \qquad (p\in\atoms,c\in\unitInterval).
\end{equation}
As formally defined later,~$\neg$ and~$\land$ respectively denote
fuzzy negation and (Zadeh) fuzzy conjunction, while $\ominus$ denotes
truncated subtraction. We include truth constants~$c\in[0,1]$ (for
some purposes, one needs to restrict to rational~$c$,
cf.~\cref{remark:numerical-parameters-size}), following Pavelka logic
(\cref{sec:prelim}). We define disjunction~$\lor$ and the fuzzy box
modality~$\Box$ by dualization as usual
($\phi \lor \psi = \lnot (\lnot \phi \land \lnot \psi)$,
$\Box \phi = \lnot \Diamond \lnot \phi$). Moreover, we can define
truncated addition of constants, $(-)\oplus c$, and constants~$c$ as
$\phi \oplus c = \lnot (\lnot \phi \ominus c)$ and
$c = (\lnot 0) \ominus (1-c)$, respectively. As usual, a
\emph{literal} is either an atom~$p$ or the negation~$\neg p$ of an
atom. The \emph{syntactic size} of a formula
is 
just the number of nodes in its syntax tree; in particular, each
constant contributes only size~$1$. (A more fine-grained size measure
taking into account the binary representation size of constants, now restricted to rationals, will not change our asymptotic size bounds; roughly speaking, numbers are subjected only to addition, so the dependence on binary representation size is only linear. See \Cref{remark:numerical-parameters-size} for more details.)
  
\subparagraph{Semantics} A \emph{fuzzy Kripke model} is a
triple~$\mathcal{M} =
\bigl(W_{\mathcal{M}},R_{\mathcal{M}},V_{\mathcal{M}}\bigr)$, where
$W_{\mathcal{M}}\neq\emptyset$ is a set of worlds,
$R_{\mathcal{M}}\colon W_{\mathcal{M}}\times
W_{\mathcal{M}}\to\unitInterval$ is a fuzzy accessibility relation,
and
$V_{\mathcal{M}}\colon\atoms\times W_{\mathcal{M}}\to\unitInterval$ is
a valuation. We will regard the designators
$W_{\mathcal{M}},R_{\mathcal{M}},V_{\mathcal{M}}$ as being implicitly
determined by~$\mathcal M$ throughout. The valuation extends
recursively to all formulae by
\begin{gather*}
  \begin{aligned}
    V_{\mathcal{M}}(0,w)
    &= 0
    &\qquad
      V_{\mathcal{M}}(\phi\ominus c,w)
    &= \max\{V_{\mathcal{M}}(\phi,w)-c,0\}
    \\
    V_{\mathcal{M}}(\lnot \phi,w)
    &= 1-V_{\mathcal{M}}(\phi,w)
    &\qquad
      V_{\mathcal{M}}(\phi\land\psi,w)
    &=
      \min\{V_{\mathcal{M}}(\phi,w),V_{\mathcal{M}}(\psi,w)\}
  \end{aligned}
  \\
  V_{\mathcal{M}}(\Diamond\phi,w)
  =\textstyle
  \sup_{v\in W_{\mathcal{M}}}
  \min\{R_{\mathcal{M}}(w,v),V_{\mathcal{M}}(\phi,v)\}.
\end{gather*}
We thus interpret conjunction~$\land$ as minimum (and hence
disjunction~$\vee$ as maximum) as recalled in \cref{sec:prelim},
while~$\ominus$ and~$\oplus$ allow limited use of the additive
structure of the unit interval, which underlies \L{}ukasiewicz
semantics. In particular, the operation $(-)\oplus~c$ is
\L{}ukasiewicz disjunction with the constant~$c$; the non-expansive
propositional base, consisting of all connectives in
\eqref{eq:grammar} except the modality, may thus be alternatively
described either as adding constant shift operators $(-)\ominus c$ to
the Zadeh base, or as restricting \L{}ukasiewicz disjunction and
conjunction~$\oplus$,~$\otimes$ to have one argument constant. The
fuzzy diamond $\Diamond$ is interpreted in analogy to the classical
two-valued modal diamond, with transition weight
$R_{\mathcal{M}}(w,v)$ and the truth value of the argument
formula~$\phi$ being combined by minimum, in agreement with the
interpretation of conjunction~$\land$. Unfolding the definition
of~$\Box$, we obtain
\begin{equation}\label{eq:box-sem}
  V_{\mathcal{M}}(\Box\phi,w)
  =\textstyle
  \inf_{v\in W_{\mathcal{M}}}
  \max\{1-R_{\mathcal{M}}(w,v),V_{\mathcal{M}}(\phi,v)\}.
\end{equation}

\subparagraph*{Sequents} We will work with a sequent calculus
(\cref{sec:sequent}) whose sequents consist of threshold judgements,
or just `judgements', comparing the truth values of formulae with
constants. It will be convenient to have notation abstracting over
comparison operators:

\begin{convention}
  Throughout, we use
  $\ineqOpSymbol \in \ineqOpSet, \lineqOpSymbol \in \lineqOpSet$, and
  $\gineqOpSymbol \in \gineqOpSet$ as placeholders for comparison
  operators. As expected, the operators $>,<$ are \emph{strict} and
  the operators $\ge,\le$ are \emph{non-strict}; we refer to this
  property as the \emph{strictness} of a comparison
  operator. Similarly, the \emph{direction} of a comparison operator is
  either left-to-right ($<,\le$) or right-to-left ($>,\ge$)
\end{convention}

\begin{defn}
  A \emph{(threshold) judgement} is an inequality~$\phi \ineqOp c$, where $\phi$
  is a formula, $\ineqOpSymbol \in\ineqOpSet$ is a comparison
  operator, and $c\in\mathbb{R}$ (we explicitly do not restrict
  to~$c\in[0,1]$). A \emph{sequent} is a finite set of judgements.
  The size of a judgement $\phi \ineqOp c$ is
  $\lvert \phi \rvert + 2$, and the size of a sequent is the sum of
  the sizes of its judgements. \emph{Satisfaction} $\models$ of
  judgements $\phi \ineqOp c$ and sequents~$\Gamma$ in a world
  $w\in W_{\mathcal{M}}$ of a fuzzy Kripke model~$\mathcal{M}$ is
  defined by
  \begin{align*}
    \mathcal{M},w\models\phi \ineqOp c & \iff\quad V_{\mathcal{M}}(\phi,w) \ineqOp c \\
    \mathcal{M},w\models\Gamma & \iff\quad \mathcal{M},w\models A \text{ for some judgement }A\in\Gamma.
  \end{align*}
  A judgement or sequent~$\Gamma$ is \emph{valid}, written
  $\models\Gamma$, if $\mathcal{M},w\models\Gamma$ for every fuzzy
  Kripke model~$\mathcal{M}$ and every world~$w\in
  W_{\mathcal{M}}$. Moreover,~$\Gamma$ is \emph{satisfiable} if there
  exist a fuzzy Kripke model~$\mathcal{M}$ and a
  world~$w\in W_{\mathcal{M}}$ such that $\mathcal{M},w\models\Gamma$.
\end{defn}

\begin{example}
  We give an example of a formula in non-expansive fuzzy modal
  logic, extending a similar formula in \cite[Example~2.2]{WildEA18} so
  as to illustrate the use of constant shifts:
	We take transition degrees in fuzzy transition systems as our accessibility interpretation. Then, the formula
	\begin{equation*}
	\Diamond\bigl((\Box 0)\oplus\tfrac15\bigr)
	\end{equation*}
	measures the degree to which a nearly deadlocked state can be
	reached in one step. In particular, $(\Box 0)(y)=1-\sup_z R(y,z)$
	is the deadlock degree of~$y$. The shift means that a state is treated as deadlocked to a higher degree than its successorship degrees ordinarily suggest, i.e. we discount the successorship values by up to~$1/5$.
\end{example}

\section{Normal Forms}\label{normalforms}
 We introduce two normalizations used in the
sequent calculus (\cref{sec:sequent}) and in the computation of
uniform interpolants (\cref{sec:uniform-interpolation}). First, we
push negations inward to obtain negation normal form, with only a
linear increase in syntactic size. We then introduce adjustments of
judgements that allow us, under suitable conditions, to prescribe
their comparison operators and numerical bounds.

For formulae~$\phi,\psi$, we write $\phi \equiv \psi$ if $\phi$ and
$\psi$ have the same truth value at every world of every model. For
judgements and sequents, $\equiv$ denotes equivalence of satisfaction
at every world of every model.
	
\begin{defn}
  A formula~$\phi$ is in \emph{negation normal form} if negations
  occur only immediately before atoms in~$\phi$. Here, the derived
  connectives $\lor$, $\oplus$, and $\Box$, as well as the truth
  constants become first-class citizens rather than
  abbreviations. That is, formulae in negation normal form are
  generated by
  \begin{equation*}
    \phi,\psi ::=
    c \mid p \mid \lnot p
    \mid \phi \land \psi
    \mid \phi \lor \psi
    \mid \phi \ominus c
    \mid \phi \oplus c
    \mid \Diamond\phi
    \mid \Box\phi, \qquad (p\in\atoms,c\in\unitInterval).
  \end{equation*}
\end{defn}
A \emph{modalized formula} is a formula of the form either
$\Diamond\phi$ or~$\Box\phi$.  The following equivalences allow us to
push negations inward:
	\begin{align*}
		\lnot\lnot\phi
		&\equiv \phi,
		&
		\lnot c
		&\equiv 1-c,
		&
		\lnot(\phi \land \psi)
		&\equiv \lnot\phi \lor \lnot\psi,
		\\
		\lnot(\phi \lor \psi)
		&\equiv \lnot\phi \land \lnot\psi,
		&
		\lnot(\phi \ominus c)
		&\equiv (\lnot\phi) \oplus c,
		&
		\lnot(\phi \oplus c)
		&\equiv (\lnot\phi) \ominus c,
		\\
		\lnot\Diamond\phi
		&\equiv \Box\lnot\phi,
		&
		\lnot\Box\phi
		&\equiv \Diamond\lnot\phi.
	\end{align*}
	Here, $1-c$ denotes the truth constant with that value. 
	
	\begin{lemma}\label{lemma:nnf}
		Every formula $\phi$ is equivalent to a formula $\psi$ in
		negation normal form with
		$\lvert\psi\rvert = O(\lvert\phi\rvert)$.
		Moreover, $\psi$ has the same atoms and modal depth as $\phi$.
	\end{lemma}
	
        \begin{defn}
          The \emph{logical complement} of a judgement $\phi\ineqOp c$ is the judgement $\theta\strictnessOp{\ineqOp} 1-c$ where $\theta$ is the negation normal form of $\neg\phi$.
        \end{defn}
        As the name indicates, the logical complement of a judgement~$\alpha$ is satisfied iff~$\alpha$ is not satisfied.
        
	Next, we introduce a way of adjusting judgements. Roughly speaking, this allows us to choose a comparison operator and a numerical bound in advance and replace a given judgement by an equivalent one with exactly this comparison operator and bound. This will allow us to combine judgements conveniently using propositional operators. More importantly, when applying the modal rules of the calculi introduced later, we will be able to adjust a judgement involving a modal formula so that removing the modality yields a judgement with a prescribed comparison operator and bound.
	
	\begin{defn}
		Let~$\phi \ineqOp c$ be a judgement, where~$c \in \unitInterval$. Let~$\ineqOpSymbol' \in \ineqOpSet$ and~$c' \in \unitInterval$. A \emph{$(\ineqOpSymbol',c')$-adjustment of~$\phi \ineqOp c$} is a formula~$\psi$ such that
		\begin{equation*}
			\phi \ineqOp c
			\quad\equiv\quad
			\psi \ineqOp' c'.
		\end{equation*}
	\end{defn}
    \noindent We now give a construction of adjustments in this sense.
    Note that this construction is partial as adjustments may fail to exist in edge cases (\Cref{remark:adjustmentConstructionUndefined}). Whenever we employ the construction later, it will be defined. 
          
    For a comparison operator~$\ineqOpSymbol$, we write $\strictnessOp{\ineqOpSymbol}$ for the comparison operator obtained by changing strictness while preserving direction, and $\directionOp{\ineqOpSymbol}$ for the comparison operator obtained by reversing direction while preserving strictness. We write $\complementOp{\ineqOpSymbol} = \directionOp{(\strictnessOp{\ineqOpSymbol})}$, so $\complementOp{\ineqOpSymbol}$ is the logical complement of~$\ineqOpSymbol$.
    For example, $\strictnessOp{>}$ is~$\geq$, $\directionOp{>}$ is~$<$, and $\complementOp{>}$ is~$\leq$.
   
    For a formula~$\phi$ and a real number~$t$, we define the
    \emph{truncated shift}~$\shiftformula{t}{\phi}$ of~$\phi$ by~$t$
    by
    \begin{equation*}
    	\shiftformula{t}{\phi}\coloneq
    	\begin{cases}
    		0 &\text{if }t\leq -1,\\
    		\phi\ominus(-t) &\text{if }-1<t<0,\\
    		\phi &\text{if }t=0,\\
    		\phi\oplus t &\text{if }0<t<1,\\
    		1 &\text{if }t\geq 1.
    	\end{cases}
    \end{equation*}
    Thus, $\shiftformula{t}{\phi}$ is a formula such that
    \begin{equation*}
    	V_{\mathcal{M}}(\shiftformula{t}{\phi},w)
    	=\min\bigl(1,\max(0,V_{\mathcal{M}}(\phi,w)+t)\bigr)
    \end{equation*}
    for every fuzzy Kripke model~$\mathcal{M}$ and every
    world~$w\in W_{\mathcal{M}}$. By straightforward computation, one
    shows:
	\begin{lemma}\label{lemma:assertionnormalization}
		Let~$\phi \ineqOp c$ be a judgement and let $\ineqOpSymbol' \in \ineqOpSet$ be a comparison operator such that~$\ineqOp$ and~$\ineqOpSymbol'$ are either both strict or both non-strict, and let~$c, c' \in \unitInterval$. Then the following construction of $\normassertion{\ineqOpSymbol'}{c'}(\phi \ineqOp c)$ yields a $(\ineqOpSymbol',c')$-adjustment of~$\phi \ineqOp c$ when it is defined:
		\begin{align*}
			\normassertion{\geq}{c'}(\phi\geq c)
			&\coloneq \shiftformula{c'-c}{\phi},
			&&\text{if }c'>0\text{ or }c=0,\\
			\normassertion{>}{c'}(\phi>c)
			&\coloneq \shiftformula{c'-c}{\phi},
			&&\text{if }c'<1\text{ or }c=1,\\
			\normassertion{\gineqOpSymbol}{c'}
			(\phi\directionOp{\gineqOp}c)
			&\coloneq
			\normassertion{\gineqOpSymbol}{c'}
			(\lnot\phi\gineqOp 1-c),\\
			\normassertion{\directionOp{\gineqOpSymbol}}{c'}
			(\phi\ineqOp c)
			&\coloneq
			\lnot\normassertion{\gineqOpSymbol}{1-c'}
			(\phi\ineqOp c).
		\end{align*}
	\end{lemma}
        \noindent
        Whenever~$\normassertion{\ineqOpSymbol'}{c'}(\phi \ineqOp c)$
        is defined, then we have
        $\lvert \normassertion{\ineqOpSymbol'}{c'}(\phi \ineqOp c)
        \rvert \leq \lvert \phi \rvert + 4$. Here each numerical
        constant counts as having size~$1$. If rational constants are
        instead represented in binary, the additional encoding length
        is instead linear in the combined encoding lengths of~$c$ and~$c'$.
	
	\begin{remark}\label{remark:adjustmentConstructionUndefined}
          Under the assumptions of
          \cref{lemma:assertionnormalization}, the construction of
          $\normassertion{\ineqOpSymbol'}{c'}(\phi \ineqOp c)$ is
          defined whenever the prescribed comparison operator and
          numerical bound do not make the target judgement trivially
          true or false. The target forms for which the construction
          may be undefined are $\psi\geq 0$ and $\psi\leq 1$, which
          are always true, and $\psi>1$ and $\psi<0$, which are always
          false.  For an always-true target, the construction is
          defined precisely when the source judgement is $\phi\geq 0$
          or $\phi\leq 1$.  Similarly, for an always-false target, it
          is defined precisely when the source judgement is $\phi>1$
          or $\phi<0$.  Undefinedness of the construction does not
          imply that no adjustment exists, but the undefined cases are
          always trivial.
	\end{remark}

	\begin{convention}
		For ease of notation, we will from here on also just write~$\normassertiondefault(\phi \ineqOp c)$ for~$\normassertion{\gineqOpSymbol}{0.5}(\phi \ineqOp c)$, where $\gineqOpSymbol\in\{>,\ge\}$ has the same strictness as~$\ineqOpSymbol$. We call this the \emph{normalization} of $\phi \ineqOp c$.
	\end{convention}
	
        \section{A Sequent Calculus For Fuzzy Modal
          Logic}\label{sec:sequent}
Next, we introduce a sequent calculus for validity in non-expansive
fuzzy modal logic, which is essentially obtained 
as the dual of the tableau calculus
for non-expansive fuzzy $\ALC$~\cite{GebhartEA25}. 

While classical propositional and first-order logic admit
Gentzen-style sequent calculi~\cite{Gentzen35}, proof systems for \L{}ukasiewicz logic often include hypersequent and labelled sequent calculi (e.g.~\cite{MetcalfeEA05,KulackaEA13}).
In the case of non-expansive fuzzy modal logic, we can, however, retain a comparatively simple, single-sequent calculus with Gentzen-style rules. This allows us to avoid much of the additional proof-theoretic machinery that is commonly used for full \L{}ukasiewicz logics.

The calculus works with sequents being finite sets of judgements
$\phi\ineqOp c$ as per \cref{sec:fuzzy-ml}. 
It also uses the \emph{coverage composite}~$\covcompDef{\gineqOpSymbol_1}{\gineqOpSymbol_2}$, which is defined by
\begin{equation*}
	\covcompDef{\gineqOpSymbol_1}{\gineqOpSymbol_2} = \begin{cases}
		> &\text{if both comparison operators are strict,}\\
		\geq &\text{otherwise.}
	\end{cases}
\end{equation*}
Equivalently, for $x,z\in\mathbb{R}$, we have $x \covcomp{\gineqOpSymbol_1}{\gineqOpSymbol_2} z \Longleftrightarrow \forall y \in \mathbb{R} . (x \gineqOp_1 y \lor y \gineqOp_2 z)$.
The rules of the calculus
are shown in \Cref{tab:seqcalc}. 
We say that a rule in \Cref{tab:seqcalc} is \emph{applicable} to, or \emph{applies} to, a sequent~$\Gamma$ if $\Gamma$ is an instance of the conclusion of the rule and the corresponding side conditions are satisfied. A sequent is then \emph{derivable} if it can be proved by finitely many rule applications, i.e., each branch eventually ends at a leaf with an empty label. Every rule in \Cref{tab:seqcalc} except the rule~$(\Box)$ is a
\emph{propositional rule}. A sequent is \emph{propositionally
  saturated} if no propositional rule applies to it.
	
The propositional rules are understood semantically as
follows. Rules~$(\text{Ax }c)$,~$(\text{Ax }p)$,
and~$(\text{Range}_\top)$ are axioms capturing valid
inequalities among constants, the fact that the truth value of any
atom~$p$ lies in the unit interval, and that the truth value of any
formula is at least~$0$, respectively. The side condition
of~$(\text{Ax }p)$ says that $c_1+c_2$ is `small enough' compared
to~$1$, e.g.~that $1\ge c_1+c_2$ in case $\gineqOp_1$ is~$\ge$ and
$\gineqOp_2$ is~$>$. The rules~$(\text{Range}_\bot)$ and~$(\text{Const}_\bot)$ are used to drop
unsatisfiable judgements demanding that the truth value of a formula
exceeds~$1$ or a constant satisfies a false judgement. Although a general weakening rule would
subsume~$(\text{Range}_\bot)$ and~$(\text{Const}_\bot)$, we include them explicitly because the
\interpolationCalculus for computing interpolants introduced later requires analogous rules that apply only to judgements of these forms and because they are, unlike the more general weakening rule, invertible. The
rules $(\lor)$,~$(\land)$ follow the semantics of the
corresponding connectives.  The side condition of the otherwise
straightforward rule~$(\ominus)$ is not strictly required for soundness but
blocks application of the rule to judgements that are trivially valid
because truth values are non-negative. This will ensure the rule is invertible and will later be relevant for the \interpolationCalculus. Symmetrically, the side
condition of rule~$(\oplus)$ for truncated addition blocks application
to judgements that are trivially unsatisfiable because truth values
are never strictly above~$1$, and as such is in fact required to
ensure soundness.

The modal rule~$(\Box)$ is analogous in spirit to the corresponding
rule in classical modal logic, but additionally accounts for the
graded accessibility relation. Its side condition
$1 \covcomp{\gineqOp}{\gineqOp_i} c+c_i$
requires $c+c_i<1$ if both comparison operators are strict, and
$c+c_i\leq 1$ otherwise. For example, suppose both comparison operators are non-strict:
If $\Box\phi\geq c$ fails at a world, then some successor has
accessibility degree $u>1-c$ and satisfies $\phi<c$.
The premise then forces one of the judgements
$\psi_i\geq c_i$ to hold at that successor. Since the side condition
ensures $u>1-c\geq c_i$, this gives us
$\Diamond\psi_i\geq c_i$ at the original world.
As in the classical case, there may be several possible applications
of~$(\Box)$ to the same sequent, one for every judgement of the form $\Box \phi \gineqOp c$. 
	
\begin{table}
  \centering
  \begin{tabular}{c}
    \toprule
    \normalsize Rules of the sequent calculus \\
    \midrule
    \addlinespace[1em]
    {\ruleSize
    $\displaystyle
    (\text{Ax }c) \ \dfrac{}{\Gamma,c \gineqOp d}\ (c \gineqOp d) 
    \qquad 
    (\text{Ax }p) \ \dfrac{}{\Gamma, p \gineqOp_1 c_1, \lnot p \gineqOp_2 c_2}\ (1 \covcomp{\gineqOp_1}{\gineqOp_2} c_1 + c_2)
    \qquad
    $
    }
    \\
    \addlinespace[1em]
    {\ruleSize
    $\displaystyle
    (\text{Const}_\bot)\quad
    \dfrac{\Gamma}{\Gamma,c \gineqOp d}
    \quad(c \complementOp{\gineqOp} d)
    \qquad
    (\text{Range}_\top) \ \dfrac{}{\Gamma, \phi \gineqOp c}\ (0 \gineqOp c)
    \qquad
    (\text{Range}_\bot) \ \dfrac{\Gamma}{\Gamma, \phi \gineqOp c}\ (c \strictnessOp{\gineqOp} 1)
    $
    }
    \\
    \addlinespace[1em]
    {\ruleSize
    $\displaystyle
    (\ominus) \ \dfrac{\Gamma, \phi \gineqOp c + d}{\Gamma, \phi \ominus d \gineqOp c}\ (c \strictnessOp{\gineqOp} 0)
    \qquad
    (\oplus) \ \dfrac{\Gamma, \phi \gineqOp c - d}{\Gamma, \phi \oplus d \gineqOp c}\ (1 \gineqOp c)
    \qquad
    (\lor) \  \dfrac{\Gamma, \phi \gineqOp c, \psi \gineqOp c}{\Gamma, \phi \lor \psi \gineqOp c} 
    $
    }
    \\
    \addlinespace[1em]
    {\ruleSize
    $\displaystyle
    (\land) \  \dfrac{\Gamma, \phi \gineqOp c \qquad \Gamma, \psi \gineqOp c}{\Gamma, \phi \land \psi \gineqOp c} 
    \qquad
    (\Box) \ \dfrac{\phi \gineqOp c, \{\psi_i \gineqOp_i c_i \mid 1 \covcomp{\gineqOp}{\gineqOp_i} c + c_i\}}{\Gamma, \Box \phi \gineqOp c, \{\Diamond \psi_i \gineqOp_i c_i \mid 1 \leq i \leq n\}}
    $
    }
    \\
    \addlinespace[0.5em]
			\bottomrule
  \end{tabular}\medskip
  
  \caption{Sequent calculus for non-expansive fuzzy modal logic; side
    conditions of rules on the right.}
  \label{tab:seqcalc}
\end{table}

The calculus handles only judgements with comparison operators
$>$ and $\geq$ and assumes that formulae are in negation normal form.
By \Cref{lemma:assertionnormalization,lemma:nnf}, we may assume
both conditions without loss of generality. This preprocessing
increases syntactic size by at most a constant factor, which we
ignore from here on.
	
\begin{theorem}\label{theorem:sequentcalccorrect}
  The sequent calculus in \Cref{tab:seqcalc} is sound and
  complete, i.e.~a sequent has a derivation if and only if it
  is valid.
\end{theorem}

\section{Uniform Interpolation of Threshold Judgements}\label{sec:uniform-interpolation}
\noindent We proceed to establish an important stepping stone towards
our main result, a generalized form of deductive uniform interpolation
where we parametrize weak consequence (`if~$\phi$ has truth value~$1$,
then~$\psi$ has truth value~$1$', \cref{sec:prelim}) with lower bounds
on truth values (`if~$\phi$ has truth value at least~$c$, then~$\psi$
has truth value at least~$d$'); that is, we construct uniform post-interpolants of threshold judgements. The proof of uniform interpolation
proper (\cref{theorem:pointwise-interpolation}) will then rely on a
symbolic treatment of lower bounds. We fix the central definitions:
\begin{defn}
  \begin{enumerate}[wide]
  \item For a judgement~$\alpha$, let $\voc(\alpha)$ (the
    \emph{vocabulary} of~$\alpha$) denote the set of all atoms
    occurring in the formula of~$\alpha$.
  \item For an atom~$p$, a judgement~$\alpha$ is \emph{$p$-free} if
    $p\notin\voc(\alpha)$.
  \item For judgements~$\alpha$ and~$\beta$, we write $\alpha \models \beta$, and say that~$\beta$ is a \emph{local consequence} of~$\alpha$, if, for every model~$\mathcal{M}$ and every world~$w\in W_{\mathcal{M}}$, we have  
    \begin{equation*}
      \mathcal{M},w\models\alpha \quad\Longrightarrow\quad \mathcal{M},w\models\beta.
    \end{equation*}
  \item Suppose that $\alpha\models\beta$. A judgement~$\gamma$ is an
    \emph{interpolant} of~$\alpha\models\beta$ if $\voc(\gamma) \subseteq \voc(\alpha)\cap\voc(\beta)$,
    \begin{equation*}
      \alpha\models\gamma,
      \qquad\text{and}\qquad
      \gamma\models\beta.
    \end{equation*}
  \item A judgement~$\exists p.\,\alpha$ is a \emph{(uniform) post-interpolant} of a judgement~$\alpha$ w.r.t.~an atom~$p$ if the following conditions hold:
    \begin{enumerate}
    \item $\exists p.\,\alpha$ is $p$-free;
    \item $\alpha\models\exists p.\,\alpha$; and
    \item for every $p$-free judgement~$\gamma$,
      \begin{equation*}
        (\alpha\models\gamma)
        \quad\Longrightarrow\quad
        (\exists p.\,\alpha\models\gamma).
      \end{equation*}
    \end{enumerate}
    Dually, a judgement~$\forall p.\,\alpha$ is a \emph{(uniform) pre-interpolant} of~$\alpha$ w.r.t.~$p$ if the following conditions hold:
    \begin{enumerate}
    \item $\forall p.\,\alpha$ is $p$-free;
    \item $\forall p.\,\alpha\models\alpha$; and
    \item for every $p$-free judgement~$\gamma$,
      \begin{equation*}
        (\gamma\models\alpha)
        \quad\Longrightarrow\quad
        (\gamma\models\forall p.\,\alpha).
      \end{equation*}
    \end{enumerate}
  \end{enumerate}
\end{defn}
\noindent (Notice that what was termed weak consequence of~$\psi$
from~$\phi$ in \cref{sec:prelim} would correspond to
$\phi\ge 1\models\psi\ge 1$ in the notation introduced above.)
\begin{remark}
  It is easy to see that if a logic admits post-interpolants for every
  judgement~$\alpha$, then it also admits pre-interpolants for every
  judgement. Indeed, one may use the quantifier duality
  $\forall p.\,\alpha \equiv \lnot \exists p.\, (\lnot \alpha)$, where
  the negation of a judgement is interpreted by replacing its
  comparison operator with the complementary one. Furthermore, a logic
  that admits either pre- or post-interpolants also admits
  interpolants of local consequences, since one may simply eliminate every atom
  outside the shared vocabulary. Furthermore, post-interpolants for
  judgements extend naturally to sequents. Given a sequent, we compute
  a post-interpolant for each of its judgements and normalize the
  resulting judgements. We then group them according to the strictness
  of their original comparison operators and, within each group,
  combine their formulae using~$\lor$. Finally, we impose~$0.5$ as a lower bound with the corresponding
  strictness. The resulting post-interpolant is therefore either a single judgement or a sequent containing at most two judgements. For sequents whose judgements are either all strict or all non-strict, normalization allows us to combine the entire sequent into a single equivalent judgement, to which the pre-interpolant construction applies. For sequents containing both strict and non-strict judgements, the existence of pre-interpolants requires its own construction.
\end{remark}
The computation of post-interpolants will run in parallel with proofs
in the sequent calculus of \cref{sec:sequent}. To aid intuition, we
will in fact work with a dualized notion:
\begin{defn}
  A judgement~$\beta$ is a \emph{$p$-free complement} of a
  sequent~$\Gamma$ if~$\beta$ is a strongest $p$-free
  judgement~$\beta$ such that~$\Gamma,\beta$ is provable. That is, $\Gamma,\beta$ is provable and, for
  every $p$-free judgement~$\gamma$,
  \begin{equation*}
    \Gamma,\gamma \text{ is provable}
    \quad\Longrightarrow\quad
    \beta \models \gamma.
  \end{equation*}
\end{defn}

\noindent It is easy to see that $p$-free complements are unique up to
logical equivalence if they exist, so we write~$\elim{p}(\Gamma)$ for
the $p$-free complement of~$\Gamma$.  Judgements whose bounds lie
outside $[0,1]$ are either valid or unsatisfiable and therefore have
constant post-interpolants.  We may thus restrict to bounds in $[0,1]$
in the following lemma, which records the intended effect that
post-interpolants are obtained straightforwardly from $p$-free
complements:

\begin{lemma}\label{lem:complement-interpolant}
  If\
  $\elim{p}(\{\normassertiondefault(\phi \complementOp{\ineqOp} c) \gineqOp 0.5\})$
  exists for a judgement~$\phi \ineqOp c$ with $c \in \unitInterval$ and an atom~$p$, then 
  \begin{equation*}
    \exists p.\, (\phi \ineqOp c) = \elim{p}(\{\normassertiondefault(\phi \complementOp{\ineqOp} c) \gineqOp 0.5\})
  \end{equation*}
  is a post-interpolant of~$\phi \ineqOp c$ w.r.t.~$p$. Here, $\gineqOpSymbol$ has the opposite strictness of $\ineqOp$.
\end{lemma}

\noindent Next, we define the labels used in the \interpolationCalculus
in~\Cref{tab:seqinterpolant}. The idea is to have a sequent~$\Gamma$
on the left-hand side, which we aim to simplify, and a
judgement~$\alpha$ on the right-hand side, which we will ultimately
prove to be logically equivalent to~$\elim{p}(\Gamma)$.
	
\begin{defn}
  A \emph{\sequentlabelname} is an expression of the form
  \begin{equation*}
    \Gamma \Downarrow_p \alpha
  \end{equation*}
  where~$\Gamma$ is a sequent and~$\alpha$ is a $p$-free judgement.
\end{defn}

\noindent The intended meaning of a \sequentlabelname
$\Gamma \Downarrow_p \alpha$ is that~$\alpha \equiv
\elim{p}(\Gamma)$. As indicated above, \Cref{tab:seqinterpolant} shows
an \interpolationCalculus for \sequentlabelnameplural. In the
rules,~$q$ ranges over all atoms, including $p$, and empty
conjunctions denote the truth constant~$1$. Notice that the rules
introduce no new atoms. The rules of the \interpolationCalculus
exactly parallel the rules of the sequent calculus of
\Cref{sec:sequent}, and we apply the same terminology (propositional
rules, saturation etc.). Rule~$(\Box)$ uses the following additional
notation: We write $\lit{p}$ for the set of literals not
mentioning~$p$:
\begin{equation*}\textstyle
  \lit{p} = \bigcup_{q \in \atoms,\ q \neq p}\{q, \lnot q\}.
\end{equation*}
Since we mean to employ the rules to compute the
$p$-free complement of a single judgement
$\phi\gineqOp c$, we can assume that all judgements in the left-hand
side sequents of configurations occurring in the computation also have
the same comparison operator~$\gineqOp$ as the original judgement,
while the right hand sides employ the dual $\strictnessOp{\gineqOp}$
(as this property is an invariant of the rules); we thus assume this
format in the rules. In particular, a propositionally saturated
sequent~$\Gamma$ to which we apply rule~$(\Box)$ can then be written
in the form
\begin{equation*}
  \Gamma = \Gamma' \cup \{\Box \phi_i \gineqOp c_i \mid 1 \leq i \leq n\} \cup \{\Diamond \psi_i \gineqOp d_i \mid 1 \leq i \leq m\}
\end{equation*}
where every judgement in~$\Gamma'$ is of the form~$\phi \gineqOp c$
with~$\phi$ a literal, and $d_1 \leq \dotsb \leq d_m$.  For
$1 \leq i \leq n$ and $1 \leq j \leq m$, respectively, we then write
\begin{equation*}
  \boxsequents{\Gamma}_i \coloneq \{\phi_i \gineqOp c_i\} \cup \{\psi_k \gineqOp d_k \mid 1 - c_i \gineqOp d_k\}
  \quad\text{and}\quad
  \diamondsequents{\Gamma}_j \coloneq \{\psi_k \gineqOp d_k \mid 1 \leq k \leq j\}.
\end{equation*}
Finally, $\boxsequentsresnormalform{\Gamma}{i}$ and
$\diamondsequentsresnormalform{\Gamma}{j}$ are defined as the following
maps from judgements to formulae:
\begin{equation*}
  \boxsequentsresnormalform{\Gamma}{i} (\beta) \coloneq \normassertiondefault(\Diamond  \normassertion{\strictnessOp{\gineqOp}}{(1-c_i)} (\beta) \strictnessOp{\gineqOp} 1-c_i)
  \qquad
  \diamondsequentsresnormalform{\Gamma}{j} (\beta) \coloneq  \normassertiondefault(\Box  \normassertion{\strictnessOp{\gineqOp}}{(1-d_j)} (\beta) \strictnessOp{\gineqOp} 1-d_j)
\end{equation*}

\begin{table}
  \centering
  \begin{tabular}{c}
    \toprule
    \normalsize Rules of the \interpolationCalculus \\
    \midrule
    \addlinespace[1em]
    {\ruleSize
    $\displaystyle
    (\text{Ax }c) \ \dfrac{}{\Gamma,c \gineqOp d \Downarrow_p 0 \strictnessOp{\gineqOp} 0.5}\ (c \gineqOp d) 
    \qquad 
    (\text{Const}_\bot)\quad
    \dfrac{\Gamma \Downarrow_p \beta}
    {\Gamma,c \gineqOp d \Downarrow_p \beta}
    \quad(c \complementOp{\gineqOp} d)
    $
    }
    \\
    \addlinespace[1em]
    {\ruleSize
    $\displaystyle
    (\text{Ax }q) \ \dfrac{}{\Gamma, q \gineqOp c_1, \lnot q \gineqOp c_2 \Downarrow_p 0 \strictnessOp{\gineqOp} 0.5}\ (1 \gineqOp c_1 + c_2)
    $
    }
    \\
    \addlinespace[1em]
    {\ruleSize
    $\displaystyle
    (\text{Range}_\top) \ \dfrac{}{\Gamma, \phi \gineqOp c \Downarrow_p 0 \strictnessOp{\gineqOp} 0.5}\ (0 \gineqOp c)
    \qquad
    (\text{Range}_\bot) \ \dfrac{\Gamma \Downarrow_p \beta}{\Gamma, \phi \gineqOp c \Downarrow_p \beta}\ (c \strictnessOp{\gineqOp} 1)
    $
    }
    \\
    \addlinespace[1em]
    {\ruleSize
    $\displaystyle
    (\ominus) \ \dfrac{\Gamma, \phi \gineqOp c + d \Downarrow_p \beta}{\Gamma, \phi \ominus d \gineqOp c \Downarrow_p \beta}\ (c \strictnessOp{\gineqOp} 0)
    \qquad
    (\oplus) \ \dfrac{\Gamma, \phi \gineqOp c - d \Downarrow_p \beta}{\Gamma, \phi \oplus d \gineqOp c \Downarrow_p \beta}\ (1 \gineqOp c)
    $
    }
    \\
    \addlinespace[1em]
    {\ruleSize
    $\displaystyle
    (\lor) \  \dfrac{\Gamma, \phi \gineqOp c, \psi \gineqOp c \Downarrow_p \beta}{\Gamma, \phi \lor \psi \gineqOp c \Downarrow_p \beta} 
    \qquad 
    (\land) \  \dfrac{\Gamma, \phi \gineqOp c \Downarrow_p \beta_1 \qquad \Gamma, \psi \gineqOp c \Downarrow_p \beta_2}{\Gamma, \phi \land \psi \gineqOp c \Downarrow_p \normassertiondefault(\beta_1) \lor \normassertiondefault(\beta_2) \strictnessOp{\gineqOp} 0.5} 
    $
    }
    \\
    \addlinespace[1em]
    {\ruleSize
    $\displaystyle
    (\Box) \ \dfrac{
    \boxsequents{\Gamma}_i \Downarrow_p \boxsequents{\beta}_i \quad (1 \leq i \leq n)
    \qquad 
    \qquad 
    \diamondsequents{\Gamma}_j \Downarrow_p \diamondsequents{\beta}_j \quad (1 \leq j \leq m)
    }
    {
    \Gamma \Downarrow_p \bigwedge_{\psi \gineqOp c \in \Gamma, \psi \in \lit{p}} \normassertiondefault(\psi \complementOp{\gineqOp} c) 
    \land 
    \bigwedge_{1 \leq i \leq n} \boxsequentsresnormalform{\Gamma}{i}(\boxsequents{\beta}_i)
    \land
    \bigwedge_{1 \leq j \leq m} \diamondsequentsresnormalform{\Gamma}{j} (\diamondsequents{\beta}_j)
    \strictnessOp{\gineqOp} 0.5
    }
    $
    }
    \\
    {\ruleSize
    $\displaystyle
    (\Gamma \text{ is propositionally saturated})
    $
    }
    \\
    \addlinespace[0.5em]
    \bottomrule
  \end{tabular}
  \medskip
  
  \caption{\InterpolationCalculus for interpolant computation eliminating~$p$}
  \label{tab:seqinterpolant}
\end{table}

The rules of \Cref{tab:seqinterpolant} can be seen as a recursive
definition of a non-deterministic function computing the right hand
side of a \sequentlabelname, in a similar spirit as in work on
interpolation in two-valued modal logic~\cite{Bilkova07}; modulo
logical equivalence, the non-determinism has no impact on the
result. The presentation paralleling the sequent calculus helps
understand both this property and the definition as a whole --
effectively, we can see in each rule how `just enough truth' is added
to the target sequent to make it derivable.

A \emph{\sequenttreename} is a finite rooted tree whose nodes are
labelled by $p$-configurations. At each node, its label and the labels
of its children form an instance of a rule of
\Cref{tab:seqinterpolant}, with the node carrying the conclusion and
its children carrying the premises. In particular, leaves correspond
to rules with no premises. Correctness of the procedure is recorded as
follows:
	
\begin{theorem}\label{theorem:interpolationcalccorrect}
  Let~$\Gamma$ be a sequent all of whose judgements use the same
  comparison operator~$\gineqOp$.  Every \sequenttreename with root
  label $\Gamma\Downarrow_p\beta$ satisfies
  $\beta\equiv\elim{p}(\Gamma)$.
\end{theorem}
\noindent By \cref{lem:complement-interpolant}, we immediately obtain
\begin{corollary}[Uniform interpolation for judgements]
  Threshold judgements in non-expansive fuzzy modal logic have
  post-interpolants.
\end{corollary}
\noindent It is immediately clear from the rules that the construction does not increase
modal depth. As a consequence, we also have uniform interpolants for threshold judgements in the propositional fragment of non-expansive fuzzy $\ALC$.

We next estimate the size of interpolants as follows:
\begin{theorem}
  Each threshold judgement in non-expansive fuzzy modal logic has a
  post-interpolant of at most exponential size.
\end{theorem}
\begin{proof}
  Let~$\alpha$ be a judgement.  We preprocess the logical complement of $\alpha$ by bringing its normalization into negation normal form, as in \Cref{sec:sequent}, obtaining an equivalent judgement
  $\alpha_0$ of size $O(\lvert\alpha\rvert)$, and put
  $N=\lvert\alpha_0\rvert$. At each leaf of a \sequenttreename, the
  right-hand side has size $O(N)$. Non-branching, non-modal rules do
  not increase the size of the right-hand side. Every other rule
  increases the size of the right-hand side by at most the sum of the
  sizes of the right-hand sides of its premises plus an overhead of
  $O(N)$.  This includes the literal conjuncts introduced by modal
  rules.  Every branch has length $O(N)$, since each recursive step
  strictly decreases the total syntactic size of the formulae in the
  left-hand sequent. Hence, a computation tree with $L$ leaves has
  $O(LN)$ nodes. Summing the local overheads and the sizes at the
  leaves therefore bounds the size of the right-hand side at the root
  by $O(LN^2)$.  Since $L$ is at most exponential in $N$, and
  $N=O(\lvert\alpha\rvert)$, the resulting interpolant has at most
  exponential size in $\lvert\alpha\rvert$.
	\end{proof}
	
	\begin{remark}[Numerical parameters and size]
		\label{remark:numerical-parameters-size}
		Thus far, we have used the syntactic size of formulae to formulate bounds on the size of interpolants. If we instead measure the size of formulae by their binary encoding lengths, we need to restrict ourselves to rational parameters, rather than arbitrary real numbers.
		We then represent these rationals by their numerators and denominators in binary.
		It is also easy to see that our construction
		preserves this rationality: new numerical parameters are obtained
		from the input parameters and the fixed constants $0$, $1/2$,
		and $1$ by addition, subtraction, minimum, and maximum. More precisely, if $D$ is the least common multiple of $2$ and all the denominators of the input, we can represent every fraction just in terms of its numerator with respect to this denominator. Addition, subtraction, minimum and maximum operations still yield rationals that can be represented in the same way. The binary encoding length of constants is thus $O(\log D+1)$, which is equal to $O(B+1)$, where $B$ is the total binary encoding length of the numerical parameters in the input.
		As a result, counting binary encodings introduces at most
		an additional factor of $O(B+1)$ in the structural size bound.
		In particular, the exponential bound on the size of interpolants remains
		exponential when measured in the total binary encoding length
		of the input.
	\end{remark}

        \section{Uniform Interpolation}
	\label{sec:pointwise-interpolation}
        As announced, we now strengthen the results of
        \cref{sec:uniform-interpolation} from uniform interpolation of
        judgements to uniform interpolation of formulae, thus
        establishing what one would usually understand as the uniform
        interpolation property. Following notation already introduced
        for the propositional case (\cref{sec:prelim}), we write
        $\models\phi\leq\psi$ for formulae $\phi$ and $\psi$ if
        $V_{\mathcal M}(\phi,w)\leq V_{\mathcal M}(\psi,w)$ at every
        world~$w$ of every model~$\mathcal M$.
	
	\begin{theorem}[Uniform interpolation]
		\label{theorem:pointwise-interpolation}
		For every formula~$\phi$ and atom~$p$, there is a formula
		$\exists p.\,\phi$ such that
		\begin{enumerate}
			\item $\voc(\exists p.\,\phi)\subseteq\voc(\phi)\setminus\{p\}$;
			\item $\models\phi\leq\exists p.\,\phi$; and
			\item for every $p$-free formula~$\psi$,
			\begin{equation*}
				\models\phi\leq\psi
				\quad\Longrightarrow\quad
				\models\exists p.\,\phi\leq\psi.
			\end{equation*}
		\end{enumerate}
	\end{theorem}
	
	\begin{remark}
		Note that \Cref{theorem:pointwise-interpolation}, once again, directly implies that non-expansive fuzzy $\ALC$ also has Craig-interpolants with respect to strong consequence, i.e. for formulae $\phi, \psi$ such that $\models \phi \leq \psi$, we obtain a Craig interpolant~$\chi$, that is, $\voc(\chi) \subseteq \voc(\phi) \cap \voc(\psi)$ and both $\models \phi \leq \chi$ and $\models \chi \leq \psi$.
	\end{remark}
	
	By duality, the formula
        $\forall p.\,\phi\coloneq\lnot\exists p.\,(\lnot\phi)$ is a
        greatest $p$-free formula below~$\phi$ in the pointwise order;
        that is, $\models\forall p.\,\phi\leq\phi$, and
        $\models\psi\leq\phi$ implies
        $\models\psi\leq\forall p.\,\phi$ for every $p$-free
        formula~$\psi$.
	
	To prove \Cref{theorem:pointwise-interpolation}, we need to abstract the \interpolationCalculus to a construction extending to symbolic thresholds. We fix the atom~$p$ as the atom we want to eliminate. We introduce a variable~$A$ serving as a placeholder for real numbers, and allow arithmetical expressions built from~$A$ on the right-hand side of judgements; in fact, only~$A$ itself and expressions $b-A$ where $b\in\mathbb R$ will occur. Such expressions evaluate to a real number when a value for~$A$ is given; we will not distinguish between an expression and its value in this sense.
	
        A \emph{\symbolicSequent} is a sequent~$\Gamma$ consisting of judgements of the form $\chi > b$, where $\chi$ is in negation normal form and $b \in \mathbb{R}$. We then write
	\begin{equation*}
		\Gamma - A \coloneq \{\chi > b - A \mid \chi > b \in \Gamma\}.
	\end{equation*}
	We now give a recursive construction of a formula~$U(\Gamma)$ obeying the invariant
	\begin{equation}\label{eq:invariantproperty}
		U(\Gamma) \geq A \equiv \elim{p}(\Gamma - A) \qquad \text{for all $0 < A \leq 1$}.
	\end{equation}
	The construction builds on the \interpolationCalculus of
        \cref{sec:uniform-interpolation} by abstracting it to symbolic
        thresholds.  Side conditions on the threshold~$A$ are encoded
        using conjunctions and disjunctions with truth constants, and
        instead of normalizing assertions to a fixed bound as in
        rules~$(\land)$ and~$(\Box)$ of \cref{tab:seqinterpolant}, we
        use truncated shifts chosen in such a way as to ensure that
        the resulting formula~$U(\Gamma)$ is independent of~$A$.  The
        symbolic threshold~$A$ thus appears only in the judgement
        $U(\Gamma)\geq A$, which ultimately represents the
        corresponding $p$-complement at each threshold.

	For $t\in\mathbb R$, we write $\clip{t}=\min(1,\max(0,t))$ for both the real number and the corresponding truth constant. The propositional cases except the literals are dealt with as follows:
	\begin{align*}
		U(\Gamma, c > b) &\coloneq U(\Gamma) \land \clip{b-c}\\
		U(\Gamma, \chi \lor \eta > b) &\coloneq U(\Gamma, \chi > b, \eta > b)\\
		U(\Gamma, \chi \land \eta > b) &\coloneq U(\Gamma, \chi > b) \lor U(\Gamma, \eta > b)\\
		U(\Gamma, \chi \ominus c > b) &\coloneq U(\Gamma, \chi > b+c) \land \clip{b}\\
		U(\Gamma,\chi\oplus c>b) &\coloneq U(\Gamma,\chi>b-c) \lor(U(\Gamma)\land\clip{b-1})
	\end{align*}
	For a \symbolicSequent~$\Gamma$ to which the above clauses have been exhaustively applied, and which thus contains only literals and modalized formulae as the left-hand sides of its judgements, we need some additional notation:
	\begin{equation*}
		C_\Gamma \coloneq \bigwedge_{\chi > b \in \Gamma} \clip{b} \land \bigwedge_{\substack{p > b \in \Gamma\\ \lnot p > d \in \Gamma}}
		\clip{\frac{b+d-1}{2}}, \qquad L_\Gamma \coloneq \bigwedge_{\substack{\ell > b \in \Gamma\\ \ell \in \lit{p}}} \shiftformula{b-1}{\lnot \ell}.
              \end{equation*}
              (with truncated shift~$\shiftformula{t}{\phi}$ defined as per \cref{normalforms}).
	These define the constant and literal parts of $U(\Gamma)$. Now, let the judgements in~$\Gamma$ with modalized formulae on the left hand side be enumerated as 
	\begin{equation*}
		\Box \phi_i > b_i\text{ for $1 \leq i \leq n$}\quad\text{and}\quad\Diamond \psi_j > d_j\text{ for $1 \leq j \leq m$ where $d_1 \leq \ldots \leq d_m$.}
	\end{equation*}
	Then we complete the definition of $U(\Gamma)$ by
	\begin{align*}
		\diamondsequentsresnormalformsymbolic{\Gamma}{j} &=  \shiftformula{d_j-1}{\Box\shiftformula{1-d_j}{U(\{\psi_k > d_k \mid 1 \leq k \leq j\})}}\\
		\boxsequentsresnormalformsymbolic{\Gamma}{i,j} & = \shiftformula{b_i-1}{\Diamond \shiftformula{1-b_i}{U(\{\phi_i > b_i\}\cup\{\psi_k > d_k \mid 1 \leq k \leq j\})}}\\
		U(\Gamma) & = C_\Gamma \land L_\Gamma \land \bigwedge_{i=1}^n \bigwedge_{j=0}^m \left(\boxsequentsresnormalformsymbolic{\Gamma}{i,j} \lor \clip{\frac{b_i + d_j -1}{2}}\right) \land \bigwedge_{j=1}^m \diamondsequentsresnormalformsymbolic{\Gamma}{j}
        \end{align*}
        where we put $d_0 = 0$ and $U(\emptyset)=1$.  Note that if
        there are no judgements with modalized formulae on the left,
        then the above clause for $U(\Gamma)$ does not contain any
        recursive calls.  Furthermore, every recursive call strictly
        decreases the syntactic size of~$\Gamma$, so the recursion
        terminates. We can also immediately see that no new atoms are
        being introduced and that the resulting formula is $p$-free.
        
        The main difficulty of the proof of \Cref{theorem:pointwise-interpolation} is showing the invariant~\eqref{eq:invariantproperty} holds. The idea is to do induction over the total syntactic size of formulae in~$\Gamma$ and compare $U(\Gamma)$ with the output of the \interpolationCalculus on $\Gamma - A$ for a symbolic threshold $A$. The details can be found in the appendix as \Cref{lemma:pointwise-invariant}.
	
	\begin{proof}[Proof of \Cref{theorem:pointwise-interpolation}]
		For a formula~$\phi$ with $\theta$ as a negation normal form of $\lnot \phi$, we put
		\begin{equation*}
			\exists p.\,\phi = U(\{\theta > 1\}).
		\end{equation*}
		Then for $A \in (0,1]$, we have $(\exists p.\,\phi) \geq A \equiv \elim{p}(\{\theta > 1 - A\})$, and $\theta > 1 - A$ is the logical complement of $\phi \geq A$. Therefore $(\exists p.\,\phi) \geq A$ is the post-interpolant of $\phi \geq A$. In particular,
		\begin{equation*}
			\phi \geq A \models (\exists p.\,\phi) \geq A.
		\end{equation*}
		And, since this holds for every positive threshold, we have $\models \phi \leq \exists p.\,\phi$. Similarly, for a $p$-free $\psi$ with $\models \phi \leq \psi$ we have $\phi \geq A \models \psi \geq A$ for every positive $A$. The post-interpolant property gives
		\begin{equation*}
			(\exists p.\,\phi) \geq A \models \psi \geq A.
		\end{equation*}
		We thus immediately have $\models \exists p.\,\phi \leq \psi$.
	\end{proof}
	\noindent As with uniform interpolation for threshold judgements, the construction of pointwise uniform interpolants
	does not increase modal depth. We thus once again obtain pointwise uniform interpolation for the propositional fragment of non-expansive fuzzy $\ALC$.

	\begin{theorem}
		The pointwise uniform interpolants $\exists p.\,\phi$ and
		$\forall p.\,\phi$ have syntactic size~$2^{O(|\phi|)}$.
	\end{theorem}
	\begin{proof}
		The proof of the size bound uses the same idea as that for the generalized form of deductive interpolants. Let $N$ be the size of the initial symbolic sequent. Once again, every recursive call strictly decreases syntactic size, so every branch of the computation tree has length $O(N)$. Furthermore, the number of total nodes in the computation tree is again in $O(LN)$, where $L$ is the number of leaves. We obtain the bound $L\leq 2^N$ for the number of leaves in the following way:
		First, we note that the propositional clauses use at most two recursive calls, while the modal clause for $n$ boxes and $m$ diamonds uses $n(m+1)+m$ recursive calls.
		Additionally, for the modal clause, each recursive call decreases the syntactic size by at least $k=n+m$, as for each call every one of the $n+m$ modal formulae is either discarded completely or retained with its outer modality removed.
		Since $n(m+1)+m < 2^k$, induction on syntactic size yields the bound~$L\leq 2^N$.
		This time, the local overhead is $O(N^2)$, because of the opposing-literal conjuncts in $C_\Gamma$ and the $n(m+1)$ modal contributions. We then obtain $|U(\Gamma)|=O(LN^3)=O(N^3 2^N)=2^{O(N)}$.
		The bound for $\forall p.\,\phi$ follows by duality.
	\end{proof}
	
\noindent	Note that rational input parameters yield rational parameters throughout: The construction uses only addition, subtraction, clipping, and division by~$2$. Thus, when replacing syntactic size by encoding length, and restricting to rationals, we still obtain an exponential bound for the size of interpolants. 
It is also clear that both this construction and the \interpolationCalculus of \Cref{sec:uniform-interpolation} yield exponential-time algorithms when restricting numerical parameters to the rationals (encoded in binary); the proofs of their size bounds give an exponential bound on the number of nodes in their computation trees and rule selection and local numerical computations take polynomial time per node. Their binary encoding lengths also remain polynomially bounded along each branch. As such, computing an interpolant takes at most exponential time in the total binary encoding length of the input.

        \section{Conclusions}

        We have established that \emph{non-expansive fuzzy modal
          logic}, which occupies an intermediate position between
        Zadeh modal logic on the one hand and \L{}ukasiewicz or
        more precisely Pavelka modal logic on the other hand, has uniform
        interpolation. Indeed, uniform interpolants have (only)
        exponential size and can be computed in exponential time, the
        same as for the two-valued correspondent (i.e.~the modal
        logic~$\mathbf K$ or the description logic $\ALC$). The
        computation of uniform interpolants is performed along proofs
        in a sequent calculus obtained by dualization of a previous
        tableau calculus~\cite{GebhartEA25}.

        Future work concerns on the one hand the extension of the
        result to more expressive logics, such as fuzzy hybrid logic
        (in the two-valued setting, uniform interpolation does extend
        to hybrid logic~\cite{JungEA26b}) or fuzzy first-order
        logic~\cite{Hajek98}. On the other hand, there is interest in
        varying the underlying type of models, e.g.~to probabilistic
        models. One would thus aim to cover logics such as
        quantitative probabilistic modal logics, in which use of the
        non-expansive propositional base is
        well-established~\cite{BreugelWorrell05}. Ideally, such
        results would be obtained generically, using methods of
        (fuzzy) coalgebraic
        logic~\cite{SchroderPattinson11,WildSchroder21}.

\bibliography{interpolation}

\newpage
\appendix 
\section{Appendix: Additional Details and Omitted Proofs}\label{appendix}

\subsection*{Details for \cref{sec:prelim}}
\subsubsection*{Failure of Craig Interpolation for Propositional Pavelka Logic}	
	Fix a set~$\atoms$ of atoms. Then the formulae of \emph{propositional Pavelka logic} are defined by the grammar:
	\begin{equation*}
		\phi, \psi \Coloneqq c \mid p \mid \lnot \phi \mid \phi \oplus \psi
	\end{equation*}
	where $c \in \unitInterval$ are constants and $p \in \atoms$ atoms.
	Their semantics is defined over the same kind of fuzzy Kripke frames:
	\begin{gather*}
		V_{\mathcal{M}}(c,w) = c 
		\qquad
		V_{\mathcal{M}}(\lnot \phi,w) = 1-V_{\mathcal{M}}(\phi,w)\\
		V_{\mathcal{M}}(\phi \oplus \psi,w) = \min\{V_{\mathcal{M}}(\phi,w) + V_{\mathcal{M}}(\psi,w),1\}.
	\end{gather*}
	We define interpolants the same way as before, but obtain that even in the propositional case, this logic does not even admit Craig interpolants for the pointwise order: Write $\phi \ominus \psi \coloneq \lnot (\lnot \phi \oplus \psi) = \max\{0, \phi - \psi\}$, $\phi \lor \psi \coloneq (\psi \ominus \phi) \oplus \phi = \max\{\phi, \psi\}$, and $\phi \land \psi = \lnot (\lnot \phi \lor \lnot \psi) \coloneq \min\{\phi, \psi\}$. Then we obtain the following counterexample:
	\begin{align*}
		\phi &\coloneq p \land (q \ominus p)\\
		\psi &\coloneq r \lor (q \ominus r)
	\end{align*} 
	We then obviously have for all valuations of $p,q,r$ that $\phi(p,q) \leq q/2 \leq \psi(q,r)$ and as such $\models \phi \leq \psi$. An interpolant $\chi$ would need to be entirely dependent on $q$ and satisfy $\phi(p,q) \leq \chi(q) \leq \psi(q,r)$. However we have:
	\begin{equation*}
		\sup_p \phi(p,q) = q/2 = \inf_r \psi(q,r),
	\end{equation*}
	so $\chi(q) = q/2$. However, since every rational Pavelka logic formula also has a continuous piecewise affine function with integer slopes as evaluation function, no formula has $q/2$ as its evaluation. As such, no Craig-interpolant exists, and consequently uniform interpolants do not exist in general.

	\begin{lemma}\label{lemma:finite-realization}
	For every finite set $\Phi$ of formulae, the set of jointly
	attainable truth-value vectors for $\Phi$ is compact.
	Moreover, every such vector is realized at the root of
	a finite tree model.
	\end{lemma}
	\begin{proof}
		We prove both claims by induction on modal
		depth. At depth zero, truth values depend continuously on
		finitely many atomic valuations.
		For the induction step, we consider the set of vectors consisting
		of a successor's accessibility degree and the truth values
		of all modal arguments occurring in $\Phi$. Its closure is
		compact by the induction hypothesis, and as such, the continuous functions
		defining the relevant modal truth values therefore attain
		their extrema on this closure. We choose one point attaining
		the required extremum for each modal subformula.
		Then by induction, the formula-value coordinates of each chosen
		point are realized in a finite tree model. We attach these
		models to a new root with the chosen accessibility degrees,
		preserving all truth values in $\Phi$.
		This construction directly bounds tree depth and branching solely
		in terms of $\Phi$. We can, of course, add edges of degree zero to transform these trees to a fixed finite tree. Then their truth-value vectors
		are the continuous image of a finite-dimensional
		unit cube, proving compactness.
	\end{proof}

\subsection*{Details for \cref{normalforms}}
        	\begin{proof}[Proof of \cref{lemma:nnf}]
		The rules trivially preserve atoms and modal depth.
		The equivalences themselves immediately follow from the semantics. Applying them from left to right eliminates double negations and pushes all remaining
		negations inward. If no rule can be applied anymore in this way, the formula has reached negation normal form. No rule duplicates a subformula, and the resulting formula has only linear syntactic size.
	\end{proof}

\subsection*{Details for \cref{sec:sequent}}
	\begin{proof}[Proof of \Cref{theorem:sequentcalccorrect}]
		The axioms~$(\text{Ax }c), (\text{Ax }p)$, and~$(\text{Range}_\top)$ each have judgements that are satisfied in any model and any world in their conclusion. 
		In the first case, it is a valid inequality involving only constants; in the second case, we have judgements which are equivalent to the statement that the atom's value is in the unit interval; in the third case, we have a statement which tells us that the value of some formula is greater than or equal to~$0$.
		As such, their conclusions are valid, and by definition, application of that rule provides us with a derivation of their conclusion.
		We now show for the other propositional rules that they are both sound and invertible, i.e. that their conclusion is valid if and only if their premise(s) are valid.
		For the rule~$(\text{Range}_\bot)$ note that $\phi \gineqOp c \strictnessOp{\gineqOp} 1$ would imply $\phi > 1$, which is never satisfied. As such validity of the conclusion is equivalent to validity of the premise.
		For the rule~$(\text{Const}_\bot)$ note that $c \gineqOp d$ is indeed false according to the side condition $c \complementOp{\gineqOp} d$ and as such, validity of the conclusion is equivalent to validity of the premise.
		For the rules~$(\ominus)$ and~$(\oplus)$, note that their conditions exactly ensure that their truncated addition or subtraction judgements in the conclusion are exactly equivalent to their respective judgement in the premise. If these conditions are not satisfied, we may of course use the rules~$(\text{Range}_\top)$ or~$(\text{Range}_\bot)$ to simplify the sequent.
		For the rule~$(\lor)$ it is immediately obvious that the premise and conclusion are logically equivalent.
		Note that for the rule~$(\land)$, validity of the conclusion is equivalent to the statement that the sequent~$\Gamma$ itself is satisfied or both~$\phi \gineqOp c$ and~$\psi \gineqOp c$ are satisfied in each world of each model. For a world in a model that satisfies~$\Gamma$ we also have that both premises are satisfied, since~$\Gamma$ is satisfied, and for a world in a model where both~$\phi \gineqOp c$ and~$\psi \gineqOp c$ are satisfied we also have that both premises are satisfied, as their respective judgements are satisfied. Conversely, if both premises are satisfied in a world in a model, then the conclusion is also trivially satisfied; either~$\Gamma$ has to be satisfied or both~$\phi \gineqOp c$ and~$\psi \gineqOp c$ have to be satisfied.
		As such, both premises are valid if and only if the conclusion is valid.
		
		For soundness of rule~$(\Box)$, suppose its conclusion is
		invalid. By \Cref{lemma:finite-realization}, it has a finite
		countermodel. Failure of the principal judgement
		$\Box\phi\gineqOp c$ is therefore witnessed by some world~$v$.
		Writing $r$ for its accessibility degree, we have
		\begin{equation*}
			V(\phi,v)\complementOp{\gineqOp}c,
			\qquad
			\begin{cases}
				r>1-c,&\text{if }\gineqOpSymbol={\geq},\\
				r\geq1-c,&\text{if }\gineqOpSymbol={>}.
			\end{cases}
		\end{equation*}
		For every diamond judgement selected for the premise, the
		coverage condition gives $r\geq c_i$, with strict inequality
		whenever $\gineqOpSymbol_i={>}$.
		Since $\Diamond\psi_i\gineqOp_i c_i$ is false at the
		distinguished world, we obtain
		$V(\psi_i,v)\complementOp{\gineqOp_i}c_i$.
		Thus $v$ falsifies the premise, proving soundness.
		
		For completeness, it remains to consider propositionally
		saturated sequents
		\begin{equation*}
			\Gamma =
			\Lambda
			\cup\{\Box\phi_i\gineqOp_i c_i\mid 1\leq i\leq n\}
			\cup\{\Diamond\psi_j\gineqOp'_j d_j\mid 1\leq j\leq m\},
		\end{equation*}
		where $\Lambda$ contains only literal judgements.
		For each boxed judgement, let
		\begin{equation*}
			J_i =
			\{j\mid 1\leq j\leq m,\;
			1\covcomp{\gineqOp_i}{\gineqOp'_j}c_i+d_j\},
			\qquad
			\Delta_i =
			\{\phi_i\gineqOp_i c_i\}
			\cup\{\psi_j\gineqOp'_j d_j\mid j\in J_i\}
		\end{equation*}
		be the corresponding modal premise.
		Suppose every $\Delta_i$ is invalid. We combine countermodels
		of these premises into a countermodel of $\Gamma$, using a
		separate successor to falsify each boxed judgement.
		Choose finite pointed countermodels $(\mathcal M_i,w_i)$
		of the $\Delta_i$, take disjoint copies, and add a new root~$w$.
		Give its edge to $w_i$ degree
		\begin{equation*}
			r_i =
			\begin{cases}
				1-c_i,
				&\text{if }\gineqOpSymbol_i={>},\\
				1-c_i+\varepsilon_i,
				&\text{if }\gineqOpSymbol_i={\geq}.
			\end{cases}
		\end{equation*}
		In the second case, choose $\varepsilon_i>0$ sufficiently
		small that $r_i\leq1$ and $r_i<d_j$ for every $j\notin J_i$.
		This is possible because saturation ensures $c_i>0$, while
		the coverage condition gives $d_j>1-c_i$ for every such~$j$.
		Retain all accessibility degrees within the components and
		set all other new accessibility degrees to zero.
		Each $w_i$ witnesses failure of its boxed judgement at~$w$.
		Every diamond judgement remains false as well: at $w_i$,
		either its argument judgement belongs to $\Delta_i$ and
		is false, or the edge degree satisfies
		\begin{equation*}
			r_i<d_j \quad\text{if }\gineqOpSymbol'_j={\geq},
			\qquad
			r_i\leq d_j \quad\text{if }\gineqOpSymbol'_j={>}.
		\end{equation*}
		Since there are finitely many chosen successors, taking
		their maximum preserves these bounds. The remaining edges,
		of degree zero, cannot satisfy a diamond judgement either,
		since no range axiom applies.
		Finally, the absence of applicable range and atomic axioms
		allows us to choose the atomic valuation at~$w$ so that all
		judgements in $\Lambda$ are false. Thus $\Gamma$ is invalid.
		For $n=0$, the same argument uses a root with all outgoing
		accessibility degrees zero.
		Consequently, every valid saturated sequent has a valid
		modal premise. Together with the invertibility of the
		non-axiomatic propositional rules and the strict decrease
		of total formula size along each branch, this proves
		completeness by induction.
	\end{proof}
	
\subsection*{Details for \cref{sec:uniform-interpolation}}

\begin{proof}[Proof of \Cref{theorem:interpolationcalccorrect}]
		We prove this by investigating each of the rules in an induction argument, showing that each rule preserves this property.
		For the axioms~$(\text{Ax }c), (\text{Ax }q)$, and~$(\text{Range}_\top)$ we can see that the left hand sequent~$\Gamma'$ itself was already valid. Thus, a $p$-free complement of~$\Gamma'$ is logically equivalent to any false statement, in our case $0 \strictnessOp{\gineqOp} 0.5$.
		For the rules~$(\text{Range}_\bot), (\text{Const}_\bot), (\ominus), (\oplus)$, and~$(\lor)$ we can also immediately see that the premise of the rule is logically equivalent to the conclusion of the rule; thus by induction we can dispatch these cases.
		For the rule~$(\land)$, notice that the right hand side
		\begin{equation*}
			\beta \coloneq \normassertiondefault(\beta_1) \lor \normassertiondefault(\beta_2) \strictnessOp{\gineqOp} 0.5
		\end{equation*}
		of the conclusion of the rule is logically equivalent to the sequent~$\beta_1, \beta_2$. By induction, we thus immediately obtain that the left hand side~$\Gamma'$ of the conclusion becomes provable by adding~$\beta_1, \beta_2$. The fact that the right hand side~$\beta$ of the conclusion is indeed a $p$-free complement of the left hand side~$\Gamma'$ is more subtle: Let~$\gamma$ be another~$p$-free judgement, which makes~$\Gamma', \gamma$ provable. Then there exists a proof of~$\Gamma', \gamma$. Without loss of generality, we can assume the first rule to be applied to~$\Gamma', \gamma$ along this proof is the rule~$(\land)$, simplifying the same term as the rule~$(\land)$ in our \sequenttreename. We can do so as the propositional rules of \Cref{tab:seqcalc} are invertible, allowing us to reorder their applications.
		Then we have:
		\begin{align*}
			\lnot \beta &\equiv \normassertiondefault(\beta_1) \lor \normassertiondefault(\beta_2) \directionOp{\gineqOp} 0.5\\
			&\equiv \lnot (\normassertiondefault(\beta_1) \lor \normassertiondefault(\beta_2)) \gineqOp 0.5\\
			&\equiv (\lnot \normassertiondefault(\beta_1)) \land (\lnot \normassertiondefault(\beta_2)) \gineqOp 0.5
		\end{align*}
		Thus, we obtain a proof for~$\lnot \beta, \gamma$ in the following way: We first apply the rule~$(\land)$ to~$\lnot \beta$, obtaining the premises
		\begin{equation*}
			\lnot \normassertiondefault(\beta_1) \gineqOp 0.5, \gamma \quad \text{and} \quad \lnot \normassertiondefault(\beta_2) \gineqOp 0.5, \gamma.
		\end{equation*}
		However, the original proof for~$\Gamma', \gamma$ contains sub-proofs for both~$\Gamma'_1, \gamma$ and~$\Gamma'_2, \gamma$, which are both branches of the first applied rule, which is the rule~$(\land)$. Note, however, that by the induction hypothesis~$\beta_1$ and~$\beta_2$ are $p$-free complements of~$\Gamma'_1$ and~$\Gamma'_2$ respectively. Their negations are exactly the judgements
		\begin{equation*}
			\lnot \normassertiondefault(\beta_1) \gineqOp 0.5 \quad \text{and} \quad \lnot \normassertiondefault(\beta_2) \gineqOp 0.5.
		\end{equation*}
		It then immediately follows that both premises are provable and thus~$\lnot \beta, \gamma$ is provable.
		The rule~$(\Box)$ requires more work: First note that after repeated application of the rule~$(\land)$ in \Cref{tab:seqcalc}, we obtain for each branch a sequent of the form~$\Gamma, \alpha$, where~$\alpha$ is either:
		\begin{align*}
			\alpha &= \normassertiondefault(\psi \complementOp{\gineqOp} c) \strictnessOp{\gineqOp} 0.5 &&\equiv \psi \complementOp{\gineqOp} c \quad &&\text{for }\psi \gineqOp c \in \Gamma, \psi \in \lit{p}\\
			\alpha &= \boxsequentsresnormalform{\Gamma}{i}(\boxsequents{\beta}_i) \strictnessOp{\gineqOp} 0.5 &&\equiv \Diamond  \normassertion{\strictnessOp{\gineqOp}}{(1-c_i)} (\boxsequents{\beta}_i) \strictnessOp{\gineqOp} 1-c_i \quad &&\text{for }1 \leq i \leq n\\
			\alpha &= \diamondsequentsresnormalform{\Gamma}{j} (\diamondsequents{\beta}_j) \strictnessOp{\gineqOp} 0.5 &&\equiv \Box  \normassertion{\strictnessOp{\gineqOp}}{(1-d_j)} (\diamondsequents{\beta}_j) \strictnessOp{\gineqOp} 1-d_j \quad &&\text{for }1 \leq j \leq m
		\end{align*}
		Applying the rules~$(\ominus)$ and $(\oplus)$ appropriately, we obtain the above equivalent judgements for $\alpha$. We will thus treat $\alpha$ as being equal to one of these forms.
		For the case~$\alpha = \psi \complementOp{\gineqOp} c$, we have, by construction, $\psi \gineqOp c \in \Gamma$, so we can apply the atomic axiom to conclude the proof for~$\Gamma, \alpha$.
		For the case~$\alpha = \Diamond \normassertion{\strictnessOp{\gineqOp}}{(1-c_i)} (\boxsequents{\beta}_i) \strictnessOp{\gineqOp} 1-c_i$, note that we may apply the rule~$(\Box)$ in \Cref{tab:seqcalc} to the judgement~$\Box \phi_i \gineqOp c_i \in \Gamma$. This then gives us the sequent~$\boxsequents{\Gamma}_i,
		\normassertion{\strictnessOp{\gineqOp}}{(1-c_i)}
		(\boxsequents{\beta}_i)
		\strictnessOp{\gineqOp} 1-c_i$, which, after simplification, is equivalent to~$\boxsequents{\Gamma}_i, \boxsequents{\beta}_i$. We conclude this case by using the induction hypothesis.
		The case~$\alpha = \Box \normassertion{\strictnessOp{\gineqOp}}{(1-d_j)} (\diamondsequents{\beta}_j) \strictnessOp{\gineqOp} 1-d_j$ is analogous to the case~$\alpha = \Diamond \normassertion{\strictnessOp{\gineqOp}}{(1-c_i)} (\boxsequents{\beta}_i) \strictnessOp{\gineqOp} 1-c_i$, except we now apply the rule~$(\Box)$ in \Cref{tab:seqcalc} to~$\alpha$ itself.
		The fact that 
		\begin{equation*}
			\beta \coloneq \bigwedge_{\psi \gineqOp c \in \Gamma, \psi \in \lit{p}} \normassertiondefault(\psi \complementOp{\gineqOp} c) 
			\land 
			\bigwedge_{1 \leq i \leq n} \boxsequentsresnormalform{\Gamma}{i}(\boxsequents{\beta}_i)
			\land
			\bigwedge_{1 \leq j \leq m} \diamondsequentsresnormalform{\Gamma}{j} (\diamondsequents{\beta}_j)
			\strictnessOp{\gineqOp} 0.5
		\end{equation*}
		is a $p$-free complement of $\Gamma$ follows by a similar argument as in the case for~$(\land)$. If $\gamma$ is another $p$-free judgement which makes $\Gamma$ provable, then we may obtain a proof for~$\lnot \beta, \gamma$ in the following way: We start by simplifying~$\lnot \beta$, until we obtain the following sequent:
		\begin{gather*}
			\{\psi \gineqOp c \mid \psi \gineqOp c \in \Gamma, \psi \in \lit{p}\},
			\{\Box \normassertion{\complementOp{\gineqOp}}{c_i} (\boxsequents{\beta}_i) \gineqOp c_i \mid 1 \leq i \leq n\},\\
			\{\Diamond \normassertion{\complementOp{\gineqOp}}{d_j} (\diamondsequents{\beta}_j) \gineqOp d_j \mid 1 \leq j \leq m\},
			\gamma
		\end{gather*}
		We may assume that propositional rules are applied before modal
		rules and that each modal application selects all eligible
		diamond judgements. 
		We then follow the proof for~$\Gamma, \gamma$ in simplifying~$\gamma$, until we either apply an axiom or the rule~$(\Box)$ in \Cref{tab:seqcalc} on each branch.
		We consider the last sequent~$\Gamma'$ of each branch before an axiom or the rule~$(\Box)$ would be applied in the original proof in the respective branch and prove that we may do the same now.
		If we originally would apply an axiom, by construction, it must use at least one judgement which emerges from simplifying~$\gamma$, and at most one judgement of~$\Gamma$. Since~$\gamma$ is $p$-free, whenever a judgement of~$\Gamma$ is used to apply an axiom, it must be a judgement with a formula in~$\lit{p}$. However, we then have the same judgement in the sequent $\Gamma'$, so we may apply the same axiom. If the axiom only uses judgements from the simplification of~$\gamma$,  we may likewise apply it. Thus we can conclude a branch by applying the same axiom as in the proof for~$\Gamma, \gamma$.
		If the rule~$(\Box)$ is applied originally, we apply it to the corresponding boxed judgement in~$\Gamma'$: if the boxed judgement emerged from simplifying~$\gamma$, we use that same judgement, while if it was the $i$-th boxed judgement in~$\Gamma$, we use $\Box\normassertion{\complementOp{\gineqOp}}{c_i} (\boxsequents{\beta}_i)\gineqOp c_i$.
		Write~$\Delta$ for the original premise and~$\Delta'$ for the new premise. Let~$\Pi$ consist of the judgements in~$\Delta$ coming from~$\gamma$. These are $p$-free and have the same comparison operator, so adjustments and disjunction yield an equivalent $p$-free judgement~$\rho$. If~$\Pi$ is empty, take $\rho$ to be a false judgement.
		
		In the first case, suppose the original application selected at least one diamond judgement from~$\Gamma$. Since all eligible diamond judgements are selected and their bounds are ordered, those coming from~$\Gamma$ form a prefix $\diamondsequents{\Gamma}_j$. Thus $\Delta=\diamondsequents{\Gamma}_j\cup\Pi$. The new premise~$\Delta'$ contains~$\Pi$ and
		\begin{equation*}
			\normassertion{\complementOp{\gineqOp}}{d_j} (\diamondsequents{\beta}_j)\gineqOp d_j,
		\end{equation*}
		which is the logical complement of~$\diamondsequents{\beta}_j$. By the induction hypothesis, $\diamondsequents{\beta}_j$ is a $p$-free complement of~$\diamondsequents{\Gamma}_j$. Since~$\Delta$ is provable and~$\Pi\equiv\rho$, we therefore obtain $\diamondsequents{\beta}_j\models\rho$. Consequently, $\lnot\diamondsequents{\beta}_j,\Pi$ is valid, and hence so is~$\Delta'$. By \Cref{theorem:sequentcalccorrect}, $\Delta'$ is provable. If no diamond judgement from~$\Gamma$ was selected, the original premise consists only of~$\Pi$ and is contained in~$\Delta'$, so provability follows directly.
		
		The second case is analogous: now $\Delta=\boxsequents{\Gamma}_i\cup\Pi$, while~$\Delta'$ contains~$\Pi$ and the logical complement of $\boxsequents{\beta}_i$. The induction hypothesis gives $\boxsequents{\beta}_i\models\rho$, so again~$\Delta'$ is valid and therefore provable.
	\end{proof}

        \subsection*{Details for \cref{sec:pointwise-interpolation}}

We prove the invariant~\eqref{eq:invariantproperty} for the construction
in \Cref{sec:pointwise-interpolation}, with the atom~$p$ fixed throughout.
We first justify changing the normalization threshold in modal steps.

\begin{lemma}\label{lemma:pointwise-modal-adjustment}
    Suppose $F\geq a\equiv H\geq0.5$, where $a\in(0,1]$.
    For $s\in(0,1]$ and $\mathsf M\in\{\Diamond,\Box\}$,
    \begin{equation*}
        \shiftformula{a-s}{\mathsf M\shiftformula{s-a}{F}}\geq a
        \quad\equiv\quad
        \shiftformula{0.5-s}{\mathsf M\shiftformula{s-0.5}{H}}\geq0.5.
    \end{equation*}
\end{lemma}
\begin{proof}
    By \Cref{lemma:assertionnormalization}, the formulae
    $F_s=\shiftformula{s-a}{F}$ and
    $H_s=\shiftformula{s-0.5}{H}$ satisfy
    $F_s\geq s\equiv H_s\geq s$.
    This equivalence is preserved by both modalities.
    For $\Box$, this follows directly from its semantics: at each
    successor, either the accessibility degree is at most $1-s$ or
    the argument must have value at least~$s$.
    For $\Diamond$, a counterexample to either direction could, by
    \Cref{lemma:finite-realization}, be realized in a finite model.
    A successor witnessing one diamond judgement would then witness
    the other, a contradiction. Applying
    \Cref{lemma:assertionnormalization} to the outer shifts proves
    the claim.
\end{proof}

\begin{lemma}\label{lemma:pointwise-invariant}
    For every \symbolicSequent~$\Gamma$, \eqref{eq:invariantproperty}
    holds simultaneously for every $A\in(0,1]$.
\end{lemma}
\begin{proof}
    We do induction over $\sum_{\chi>b\in\Gamma}|\chi|$, comparing
    $U(\Gamma)\geq A$ with the output of the \interpolationCalculus
    on $\Gamma-A$. By \Cref{theorem:interpolationcalccorrect}, that
    output is equivalent to $\elim{p}(\Gamma-A)$.
    The case $\Gamma=\emptyset$ is immediate from $U(\emptyset)=1$.

    We use two observations. First, for $0<A\leq1$ and $t\in\mathbb R$,
    \begin{align*}
        (F\lor\clip{t})\geq A
        &\quad\Longleftrightarrow\quad t\geq A\ \text{or}\ F\geq A,\\
        (F\land\clip{t})\geq A
        &\quad\Longleftrightarrow\quad t\geq A\ \text{and}\ F\geq A.
    \end{align*}
    Second, the induction hypothesis gives, for smaller sequents
    $\Delta\subseteq\Sigma$,
    \begin{equation}\label{eq:pointwise-monotonicity}
        \models U(\Sigma)\leq U(\Delta).
    \end{equation}
    Indeed, $(\Delta-A),U(\Delta)\geq A$ is valid, hence so is
    $(\Sigma-A),U(\Delta)\geq A$. The strongest-complement property
    for $\Sigma-A$ gives $U(\Sigma)\geq A\models U(\Delta)\geq A$
    for every positive~$A$, proving the inequality.

    \emph{Propositional clauses.}
    Write $\Gamma=\Delta\cup\{\chi>b\}$ for the selected judgement.
    If $\chi=c$ is constant, its shifted judgement is true exactly
    when $A>b-c$; otherwise it can be deleted. The factor
    $\clip{b-c}$ gives precisely these cases.
    The clauses for $\lor$ and $\land$ follow directly from the
    corresponding rules of the \interpolationCalculus and the
    induction hypothesis.
    For subtraction, $\eta\ominus c>b-A$ is automatically true
    when $A>b$, and otherwise equivalent to $\eta>b+c-A$.
    This gives the factor~$\clip{b}$ and the recursive call:
    whenever $A>b$, the factor~$\clip{b}$ makes
    $U(\Gamma)\geq A$ false. Otherwise, $\clip{b}\geq A$
    is trivially true, so the output inequality reduces to
    $U(\Delta,\eta>b+c)\geq A$, whose correctness follows
    from the induction hypothesis.
    For addition, if $A>b-1$, the judgement $\eta\oplus c>b-A$ is equivalent to $\eta>b-c-A$.
    Since $A>\clip{b-1}$, the second disjunct cannot reach~$A$,
    so the output inequality reduces to $U(\Delta,\eta>b-c)\geq A$, whose correctness follows from the induction hypothesis.
    If $b-1 \geq A$, the principal judgement is impossible and can be discarded. The second disjunct then satisfies the inequality with~$A$ exactly when $U(\Delta)\geq A$ does, while the first disjunct is redundant by \eqref{eq:pointwise-monotonicity}.

    \emph{Range and literal contributions.}
    Suppose now that $\Gamma$ contains only literal and modal formulae.
    If $A > b$ for some bound $b$, the corresponding judgement in $\Gamma-A$
    is valid, and its factor $\clip{b}$ in $C_\Gamma$ makes
    $U(\Gamma)\geq A$ false, as required.
    We may therefore assume that all bounds are at least~$A$.
    Judgements with $b-A\geq 1$ are impossible and are therefore
    deleted by the range rule in the original \interpolationCalculus
    on~$\Gamma-A$. Our construction retains them, so we must check
    that their contributions impose no additional condition on
    $U(\Gamma)\geq A$. Their literal and modal contributions are automatically at least $\clip{b-1}\geq A$ because of their outer shifts, so they impose
    no additional condition on $U(\Gamma)\geq A$.
    For a pair of opposing $p$-literals, the corresponding factor
    in $C_\Gamma$ is $\clip{(b+d-1)/2}$. If one of their shifted judgements is impossible, i.e. $b-A \geq 1$, then $b \geq A+1$. Since the other bound
    satisfies $d\geq A$ (by the assumption that all bounds are at least~$A$), we obtain
    \begin{equation*}
    	\frac{b+d-1}{2} \geq \frac{(A+1)+A-1}{2} = A.
    \end{equation*}
    Clipping preserves this inequality because $A \leq 1$, so this factor likewise imposes no additional condition on $U(\Gamma)\geq A$.
    Impossible diamonds occur at the end of the ordered list;
    a box prefix containing one has $(b_i+d_j-1)/2\geq A$ by the same argument as for opposing $p$-literals and therefore imposes no condition. Hence these judgements have no effect on the inequality, and the remaining bounds in $\Gamma-A$ lie in $[0,1)$.

    For opposing $p$-literals, the atomic axiom applies exactly when
    $1>(b-A)+(d-A)$, or equivalently $A>(b+d-1)/2$.
    The corresponding factor in $C_\Gamma$ makes the output
    inequality false exactly in this case.
    For every retained literal~$\ell$, its contribution satisfies
    \begin{equation*}
        \shiftformula{b-1}{\lnot\ell}\geq A
        \quad\equiv\quad \ell\leq b-A.
    \end{equation*}
    Thus $L_\Gamma\geq A$ is exactly the conjunction of the retained literal complements. If these literals form an atomic axiom, this conjunction is false. Otherwise all factors of $C_\Gamma$ satisfy their inequalities with~$A$.

    \emph{Modal contributions.}
    We may now assume that no range or atomic axiom applies.
    A diamond~$j$ is eligible for a box~$i$ precisely when
    \begin{equation*}
        1-(b_i-A) > d_j-A
        \quad\Longleftrightarrow\quad A>\frac{b_i+d_j-1}{2}.
    \end{equation*}
    Thus the floor in the box contribution activates exactly the
    eligible prefixes: whenever the prefix is not eligible, the
    corresponding conjunct automatically satisfies the inequality
    with~$A$. Otherwise, the floor is below~$A$, so this inequality
    reduces to $\boxsequentsresnormalformsymbolic{\Gamma}{i,j}\geq A$. 
    The empty prefix $j=0$ is also active: $d_0=0$ and $b_i-A < 1$ imply $(b_i-1)/2 < A$ and as such its modal contribution must itself reach $A$.
    
    If $j_*$ is the largest eligible index, or zero when none exists,
    \eqref{eq:pointwise-monotonicity} and monotonicity of shifts and
    modalities give
    \begin{equation*}
        \models\boxsequentsresnormalformsymbolic{\Gamma}{i,j_*}
        \leq\boxsequentsresnormalformsymbolic{\Gamma}{i,j}
        \qquad(0\leq j\leq j_*).
    \end{equation*}
    Hence only the maximal eligible prefix contributes a condition,
    exactly as in the original \interpolationCalculus.

    For each surviving modal contribution, let~$\Delta$ be its
    recursive premise and $H\geq 0.5$ the original calculus's output
    on $\Delta-A$. By induction, $U(\Delta)\geq A\equiv H\geq 0.5$.
    Put $s=1-b+A\in(0,1]$, where $b$ is the corresponding box or
    diamond bound. Since $A-s=b-1$,
    \Cref{lemma:pointwise-modal-adjustment} gives
    \begin{equation*}
        \shiftformula{b-1}{\mathsf M\shiftformula{1-b}{U(\Delta)}}\geq A
        \quad\equiv\quad
        \shiftformula{0.5-s}{\mathsf M\shiftformula{s-0.5}{H}}\geq0.5.
    \end{equation*}
    The right-hand side is precisely the original modal adjustment,
    with $\mathsf{M}=\Diamond$ for a box and $\mathsf{M}=\Box$ for a
    diamond. Finally, conjunction combines these judgements and
    the literal judgements in both constructions. This proves
    \eqref{eq:invariantproperty} simultaneously for all positive~$A$.
\end{proof}

\end{document}